\documentclass[12pt]{article}
\usepackage{amsmath,amsthm,amssymb}
\usepackage{graphicx}
\usepackage[authoryear,round]{natbib}
\usepackage[breaklinks,hidelinks]{hyperref}
\usepackage[margin=1in]{geometry}

\newtheorem{theorem}{Theorem}[]

\newtheorem{corollary}[theorem]{Corollary}

\newtheorem{remark1}[theorem]{Remark}

\DeclareMathOperator{\E}{\mathop{}\mathbb{E}}
\def\epsilon{\varepsilon}

\title{Calibration of P-values for calibration and for deviation
       of a subpopulation from the full population}
\author{Mark Tygert\\
        Fundamental Artificial Intelligence Research, Meta Platforms, Inc.\\
        786 Coleman Ave.\ Apt.\ L, Menlo Park, CA 94025-2440\\
        \url{http://tygert.com}\\
        {\tt mark@tygert.com}}

\begin{document}

\maketitle

\begin{abstract}
The author's recent research papers,
``Cumulative deviation of a subpopulation from the full population'' and
``A graphical method of cumulative differences between two subpopulations''
(both published in volume 8 of Springer's open-access {\it Journal of Big Data}
during 2021), propose graphical methods and summary statistics,
without extensively calibrating formal significance tests.
The summary metrics and methods can measure the calibration
of probabilistic predictions and can assess differences
in responses between a subpopulation and the full population while controlling
for a covariate or score via conditioning on it.
These recently published papers construct significance tests based
on the scalar summary statistics, but only sketch how to calibrate
the attained significance levels (also known as ``P-values'') for the tests.
The present article reviews and synthesizes work spanning many decades
in order to detail how to calibrate the P-values.
The present paper presents computationally efficient, easily implemented
numerical methods for evaluating properly calibrated P-values,
together with rigorous mathematical proofs guaranteeing their accuracy,
and illustrates and validates the methods with open-source software
and numerical examples.
\end{abstract}


\section{Introduction}

Two basic problems in statistics are (1) checking calibration
of probabilistic predictions such that any event predicted
to happen, say, $x$ percent of the time actually occurs $x$ percent of the time
and (2) assessing the deviation of a subpopulation
from the full population while conditioning on a specified covariate or score
(``conditioning on'' is also known as ``controlling for,''
and involves comparing only individuals whose values for the covariate or score
are similar or otherwise match up).
In fact, the first problem can be viewed as a special case
of the second problem by requiring the expected response of the full population
to be equal to the predicted probability, so that the deviation
of the subpopulation from the full population is simply the deviation
from the probabilities. In all cases, the data consists of observations
of responses paired with scores (and weights, in the case of weighted samples).
In the first case, the scores are the predicted probabilities;
in the second case, the scores are the values of the specified covariate
(which could be probabilistic predictions, too).
In the social and biomedical sciences, controlling for income or age is common.

Recent work of~\citet{tygert_full} and~\citet{tygert_two} proposes metrics
for (inter alia) measuring miscalibration or deviation of a subpopulation
from the full population, reviewed in Subsection~\ref{tygertsummary} below.
The present paper develops methods for converting the values of such metrics
into properly calibrated attained significance levels
(also known as ``P-values''), deriving the cumulative distribution functions
for the metrics under the null hypothesis of no deviation between
the subpopulation and the full population (or of perfect calibration
in the underlying subpopulation). As reviewed below,
the works of~\citet{delgado}, \citet{diebolt}, and~\citet{stute}
prove that the estimates at finite sample sizes converge reasonably rapidly
to the limiting asymptotic values in most settings encountered in practice,
as confirmed in the numerical experiments presented below.
Figures~\ref{kuc} and~\ref{ksc} below illustrate the rapid convergence.
The metrics discussed in the present paper are very similar
to those of~\citet{kuiper} and of~\citet{kolmogorov} and~\citet{smirnov}.

The earlier works broke ties in the scores at random,
randomly ordering observations corresponding to exactly the same score.
Subsection~\ref{reduction} below proposes an alternative that avoids
any randomization (though randomization does simplify the analysis a bit).
Subsection~\ref{reduction}'s new approach may be of interest beyond just
for the calibration of P-values.

The present paper carefully elaborates on widely deployed prior work,
elucidating many details that earlier publications omitted.
The elaboration is for the convenience and reference of the reader;
the reader undoubtedly could derive most of the results presented below,
but is welcome to spare the extensive effort required by instead leveraging
the present paper and the associated open-source software.
The presentation below provides full proofs that earlier publications omitted,
and also summarizes everything required to solve the problems posed here,
rather than requiring the reader to traverse literature
that spans many decades and disciplines.
The present paper is a response to many requests for pulling together
everything into a comprehensive, convenient, reasonably elementary exposition,
as well as elaborating the simple approach of Subsection~\ref{reduction}
that not every (any?) end-user had realized was possible.
In particular, Subsection~\ref{tygertsummary} below briefly reviews
the cumulative methods of~\citet{tygert_full}
for assessing the deviation of a subpopulation from the full population;
readers unfamiliar with that approach may wish to start
with the full paper of~\citet{tygert_full} or the summary
in Subsection~\ref{tygertsummary} below.

The remainder of the paper has the following structure:
Section~\ref{methods} presents the main methods,
Section~\ref{results} validates and illustrates the methods
via numerical examples,\footnote{Software in Python 3 that implements
the methods and automatically reproduces the numerical results
(including the figures) is available
at \url{https://github.com/facebookresearch/cdeets}}
and Section~\ref{conclusion} briefly discusses
the results and draws conclusions.

\section{Methods}
\label{methods}

The present section details the methodology of the present paper.
Subsection~\ref{distributions} provides computationally efficient formulae
for evaluating the cumulative distribution functions
of the range and of the maximum absolute value of the standard Brownian motion
over the unit interval $[0, 1]$, based on the works of~\citet{feller}
and~\citet{darling-siegert}.
Subsection~\ref{tygertsummary} reviews the methods
of~\citet{tygert_full} for assessing deviation of a subpopulation
from the full population and for assessing calibration
of probabilistic predictions, introducing a graphical method
along with two statistics which summarize the graph as scalars.
Next, Subsection~\ref{calibration} shows how to use
the numerical methods of Subsection~\ref{distributions}
to calculate attained significance levels (P-values)
for the scalar summary statistics introduced in Subsection~\ref{tygertsummary},
based on the works of~\citet{delgado}, \citet{diebolt}, and~\citet{stute}.
Finally, Subsection~\ref{reduction} explains an alternative to breaking ties
in the covariates or scores at random
(randomization does simplify the analysis slightly,
but avoiding randomization altogether is possible, too).
Readers unfamiliar with the work of~\citet{tygert_full} or~\citet{tygert_two}
might want to skip to Subsection~\ref{tygertsummary} at first; however,
readers interested mainly in the numerical methods might want to start instead
with Subsection~\ref{distributions}.

\subsection{Distributions of the range and maximum absolute value
of Brownian motion}
\label{distributions}

This subsection presents series for the cumulative distribution functions
of the range and maximum absolute value of the standard Brownian motion
over the unit interval $[0, 1]$.
The terms in the series consist entirely of elementary functions that are easy
to program (as implemented in the codes mentioned in Section~\ref{results}).
The series converge rapidly and the present subsection proves rigorous bounds
on the numbers of terms required to attain a specified accuracy.
Subsubsection~\ref{rangeBrownian} gives the results
for the range of the standard Brownian motion ---
see especially Theorems~\ref{maincomborange} and~\ref{mainboundrange};
Subsubsection~\ref{maxBrownian} gives the results
for the maximum absolute value ---
see Theorems~\ref{maincombomax} and~\ref{mainboundmax}.

\subsubsection{Range of the standard Brownian motion}
\label{rangeBrownian}

This subsubsection presents Theorems~\ref{maincomborange}
and~\ref{mainboundrange}, enabling easy, rapid computation
of the cumulative distribution function for the range
(the maximum minus the minimum) of the standard Brownian motion
over the unit interval $[0, 1]$.

We define the series
\begin{equation}
\label{cdf}
F(x) = \sum_{k=1}^\infty
\left( \frac{8}{x^2} + \frac{8}{(2k-1)^2 \pi^2} \right)
\exp\left(-\frac{(2k-1)^2 \pi^2}{2x^2}\right)
\end{equation}
for any positive real number $x$.
The following theorem exhibits $F$ to be the cumulative distribution function
associated with the probability density function
of Formulae~3.6--3.8 of~\citet{feller};
Theorem~\ref{pdftheorem} below reviews those formulae.

\begin{theorem}
\label{cdftheorem}
Suppose that $F$ is the series defined in~(\ref{cdf}).
Then,
\begin{equation}
F(x) = \int_0^x f(y) \, dy
\end{equation}
for any positive real number $x$,
where
\begin{equation}
\label{pdf}
f(x) = \sqrt{\frac{2}{\pi x^2}} \;\cdot\;
\frac{\partial G}{\partial x}\left(\frac{x}{2}\right),
\end{equation}
with
\begin{equation}
\label{G}
G(x) = \frac{\sqrt{2\pi}}{x} \sum_{k=1}^\infty
\exp\left( -\frac{(2k-1)^2 \pi^2}{8x^2} \right).
\end{equation}
\end{theorem}

\begin{proof}
Clearly
$\lim_{x \to 0^+} F(x) = 0 = \lim_{x \to 0^+} \int_0^x f(y) \, dy$,
so we need only show that $\frac{\partial F}{\partial x} = f$.

Differentiating~(\ref{G}) yields
\begin{equation}
\sqrt{\frac{2}{\pi}} \;\cdot\; \frac{\partial G}{\partial x}
= \sum_{k=1}^\infty \left(
\frac{2}{x} \frac{(2k-1)^2 \pi^2}{4x^3} - \frac{2}{x^2} \right)
\exp\left( -\frac{(2k-1)^2 \pi^2}{8x^2} \right),
\end{equation}
which when combined with~(\ref{pdf}) yields
\begin{equation}
\label{first}
f(x) = \sum_{k=1}^\infty \left(
\frac{8 (2k-1)^2 \pi^2}{x^5} - \frac{8}{x^3} \right)
\exp\left( -\frac{(2k-1)^2 \pi^2}{2x^2} \right).
\end{equation}

Differentiating~(\ref{cdf}) yields
\begin{equation}
\label{second}
\frac{\partial F}{\partial x} = \sum_{k=1}^\infty
\left[ \left( \frac{8}{x^2} + \frac{8}{(2k-1)^2 \pi^2} \right)
\left( \frac{(2k-1)^2 \pi^2}{x^3} \right) - \frac{16}{x^3} \right]
\exp\left(-\frac{(2k-1)^2 \pi^2}{2x^2}\right)
\end{equation}

The right-hand sides of~(\ref{first}) and~(\ref{second}) are equal,
completing the proof.
\end{proof}

Formulae~3.6--3.8 of~\citet{feller} state the following theorem,
though Formula~3.6 of~\citet{feller} is missing a factor of $1/\sqrt{t}$.
\begin{theorem}
\label{pdftheorem}
The probability density function for the range
of the standard Brownian motion over the unit interval $[0, 1]$ is given
by Formula~(\ref{pdf}). (The range is the maximum minus the minimum.)
\end{theorem}

Combining Theorems~\ref{cdftheorem} and~\ref{pdftheorem} yields
the following theorem.
\begin{theorem}
\label{maincomborange}
The cumulative distribution function for the range (the maximum minus
the minimum) of the standard Brownian motion over the unit interval $[0, 1]$
is given by Formula~(\ref{cdf}).
\end{theorem}

The following theorem upper-bounds the tail of the series
for $F$ defined in~(\ref{cdf}).
\begin{theorem}
\label{mainboundrange}
Suppose that $n$ is a positive integer.
Then, the tail of the series for $F$ defined in~(\ref{cdf}) satisfies
\begin{equation}
\label{bound}
\sum_{k=n+1}^\infty
\left( \frac{8}{x^2} + \frac{8}{(2k-1)^2 \pi^2} \right)
\exp\left(-\frac{(2k-1)^2 \pi^2}{2x^2}\right)
< \frac{4}{\sqrt{2\pi}} \left( \frac{1}{x} + \frac{x}{\pi^2} \right)
\exp\left( -\frac{(2n-1)^2 \pi^2}{2x^2} \right)
\end{equation}
for any positive real number $x$.
If $\epsilon$ is a positive real number less than 1 and
\begin{equation}
\label{kuterms}
n \ge \frac{1}{2} + \frac{x}{\pi\sqrt{2}}
\sqrt{\ln\left( \frac{4}{\epsilon \sqrt{2 \pi}}
\left( \frac{1}{x} + \frac{x}{\pi^2} \right) \right)},
\end{equation}
then the right-hand side of~(\ref{bound}) is at most $\epsilon$.
\end{theorem}

\begin{proof}
Clearly,
\begin{equation}
\label{initial}
\sum_{k=n+1}^\infty
\!\left( \frac{8}{x^2} + \frac{8}{(2k-1)^2 \pi^2} \right)
\exp\left(-\frac{(2k-1)^2 \pi^2}{2x^2}\right)
< \left( \frac{8}{x^2} + \frac{8}{\pi^2} \right) \sum_{k=n+1}^\infty
\!\exp\left(-\frac{(2k-1)^2 \pi^2}{2x^2}\right)\!,
\end{equation}
\begin{equation}
\sum_{k=n+1}^\infty \exp\left(-\frac{(2k-1)^2 \pi^2}{2x^2}\right)
< \int_n^\infty \exp\left(-\frac{(2t-1)^2 \pi^2}{2x^2}\right) \, dt,
\end{equation}
\begin{equation}
\int_n^\infty \exp\left(-\frac{(2t-1)^2 \pi^2}{2x^2}\right) \, dt
= \frac{x}{\pi \sqrt{2}}
\int_{((2n-1) \pi)/(x\sqrt{2})}^\infty \exp(-u^2) \, du,
\end{equation}
and
\begin{equation}
\label{last}
\int_{\frac{(2n-1) \pi}{x\sqrt{2}}}^\infty \exp(-u^2) \, du
\le \exp\left( -\frac{(2n-1)^2 \pi^2}{2x^2} \right)
\int_0^\infty \!\exp(-u^2) \, du
= \frac{\sqrt{\pi}}{2} \exp\left( -\frac{(2n-1)^2 \pi^2}{2x^2} \right)\!.
\end{equation}
Combining~(\ref{initial})--(\ref{last}) yields~(\ref{bound}).
\end{proof}

\subsubsection{Maximum absolute value of the standard Brownian motion}
\label{maxBrownian}

This subsubsection presents Theorems~\ref{maincombomax}
and~\ref{mainboundmax}, enabling easy, rapid computation
of the cumulative distribution function for the maximum
of the absolute value of the standard Brownian motion
over the unit interval $[0, 1]$.

The following theorem states Formulae~3.8 and~5.2
of~\citet{darling-siegert}; see also the displayed formula immediately before
Formula~5.2 of~\citet{darling-siegert},
or Formula~2.22 of~\citet{ciesielski-taylor}
and the sentence of~\citet{ciesielski-taylor} immediately following.
\begin{theorem}
\label{maincombomax}
The cumulative distribution function for the maximum of the absolute value
of the standard Brownian motion over the unit interval $[0, 1]$ is
\begin{equation}
\label{kscdf}
D(x) = \frac{4}{\pi} \sum_{k=1}^\infty \frac{(-1)^{k-1}}{2k-1}
\exp\left(-\frac{(2k-1)^2 \pi^2}{8x^2}\right)
\end{equation}
for any positive real number $x$.
\end{theorem}

The following theorem follows from the Leibniz bound on the tail
of an alternating series for which the absolute values of the terms
in the series decrease monotonically to zero (namely, the absolute value
of the leading term of the tail is an upper bound on the absolute value
of the tail; the bound in Theorem~\ref{mainboundmax} would also be valid
if the summation started from $n$ rather than $n+1$).
\begin{theorem}
\label{mainboundmax}
Suppose that $n$ is a positive integer.
Then, the tail of the series for $D$ defined in~(\ref{kscdf}) satisfies
\begin{equation}
\label{altbound}
\left| \frac{4}{\pi} \sum_{k=n+1}^\infty \frac{(-1)^{k-1}}{2k-1}
\exp\left(-\frac{(2k-1)^2 \pi^2}{8x^2}\right) \right|
< \frac{4}{\pi} \exp\left(-\frac{(2n-1)^2 \pi^2}{8x^2}\right)
\end{equation}
for any positive real number $x$.
If $\epsilon$ is a positive real number less than 1 and
\begin{equation}
\label{ksterms}
n \ge \frac{1}{2}
+ \frac{x\sqrt{2}}{\pi} \sqrt{\ln\left(\frac{4}{\pi \epsilon}\right)},
\end{equation}
then the right-hand side of~(\ref{altbound}) is at most $\epsilon$.
\end{theorem}

\subsection{Calibration and deviation of a subpopulation
from the full population}
\label{tygertsummary}

This subsection summarizes methods of~\citet{tygert_full}
for assessing deviation of a sub\-population from the full population
and for assessing the calibration of probabilistic predictions.
The primary goal of this subsection is to introduce the Kolmogorov-Smirnov
and Kuiper metrics, as well as factors suitable for normalizing them
so as to facilitate evaluation of attained significance levels (P-values).

We consider $n$ real numbers $S_1$,~$S_2$, \dots, $S_n$
known as ``scores'' (or sometimes as ``predicted probabilities''
when calibrating probabilistic predictions),
each paired with a real-valued ``response,'' $R_1$,~$R_2$, \dots, $R_n$,
as well as a positive ``weight,'' $W_1$,~$W_2$, \dots, $W_n$;
we view the scores $S_1$,~$S_2$, \dots, $S_n$
and weights $W_1$,~$W_2$, \dots, $W_n$ as given, not random,
while we view the responses $R_1$,~$R_2$, \dots, $R_n$ as random.
We assume throughout that all responses are stochastically independent
(allowing dependence among the responses would be far beyond the scope
of the present paper).
Without loss of generality, we assume that $S_1 < S_2 < \dots < S_n$
(perturbing the original scores slightly in order to ensure their uniqueness,
if necessary).
We consider also a given function $r$ which returns the expected response
averaged over the full population at any specified score $s$;
that is, $r(s)$ is the expected value of the response for all members
of the full population whose score is $s$.
When assessing the calibration of probabilistic predictions,
the score $s$ is a predicted probability and the expected response $r(s)$
is supposed to match the prediction, $s$; hence, $r(s) = s$
when assessing calibration.

In order to gauge deviation of the observed responses
$R_1$,~$R_2$, \dots, $R_n$ from the given function $r$,
we construct the sequence of cumulative differences
\begin{equation}
B_{\ell} = \frac{\sum_{k=1}^{\ell} (R_k - r(S_k)) \, W_k}{\sum_{k=1}^n W_k}
\end{equation}
for $\ell = 1$,~$2$, \dots, $n$.
We also construct the sequence of cumulative weights
\begin{equation}
A_{\ell} = \frac{\sum_{k=1}^{\ell} W_k}{\sum_{k=1}^n W_k}
\end{equation}
for $\ell = 1$,~$2$, \dots, $n$.
Figures~\ref{la}, \ref{stanislaus}, and~\ref{napa}
of Subsection~\ref{analysis} below plot $B_1$,~$B_2$, \dots, $B_n$
versus $A_1$,~$A_2$, \dots, $A_n$ for some numerical examples.
A simple calculation of~\citet{tygert_full} shows that the expected slope
of the line graph of $B_1$,~$B_2$, \dots, $B_n$ versus 
$A_1$,~$A_2$, \dots, $A_n$ from $A_{k-1}$ to $A_k$ is $\E[R_k] - r(S_k)$;
that is, the expected slope is simply the deviation of the expected response
from the full population's, in a graph for which $A_1$,~$A_2$, \dots, $A_n$
are the abscissae (the horizontal coordinates) and $B_1$,~$B_2$, \dots, $B_n$
are the ordinates (the vertical coordinates).
Thus, steep slope over a long range indicates
significant weighted average deviation over that range.
Indeed, the slope of the secant line connecting two points on the graph becomes
the weighted average deviation over the long range of scores
between those points.

In particular, absence of significant deviation results
in a flat graph that is nearly horizontal.
Two metrics which measure deviations away from 0
(thus characterizing ``goodness-of-fit'') are the maximum absolute value
\begin{equation}
\label{kolmogorov-smirnov}
G = \max_{1 \le k \le n} |B_k|
\end{equation}
and the range (the maximum value minus the minimum value)
\begin{equation}
\label{kuiper}
H = \max_{0 \le k \le n} B_k - \min_{0 \le k \le n} B_k,
\end{equation}
where $B_0 = 0$; Remark~1 of~\citet{tygert_full} justifies including $B_0 = 0$,
a justification analogous to why~\citet{kuiper} introduced
an analogous statistic decades earlier in a related context.
The absolute value of the total deviation
$\sum_{k \in I} (R_k - r(S_k)) \, W_k \Big/ \sum_{k=1}^n W_k$
over any interval $I$ of indices is less than or equal to $H$; indeed,
\begin{equation}
H = \max_I \left| \frac{\sum_{k \in I} (R_k - r(S_k)) \, W_k}
                       {\sum_{k=1}^n W_k} \right|,
\end{equation}
where the maximum is over every interval $I$ of indices.
The statistic $G$ is due to~\citet{kolmogorov} and~\citet{smirnov};
$H$ is due to~\citet{kuiper}.

Under the null hypothesis that the response at every score $s$
is an independent Bernoulli variate taking the value $1$
with probability $r(s)$ and the value $0$ with probability $1-r(s)$,
calibrating attained significance levels (P-values) for these statistics
involves normalization by
\begin{equation}
\label{Bernoulli}
\sigma = \frac{\sqrt{\sum_{k=1}^n r(S_k) \cdot (1-r(S_k)) \cdot (W_k)^2}}
              {\sum_{k=1}^n W_k};
\end{equation}
of course, such a null hypothesis can be appropriate
only when each $R_k$ is either 0 or 1, for each $k = 1$,~$2$, \dots, $n$.
More generally, under a null hypothesis for which the response at score $s$
is expected to have a variance $v(s)$ centered around $r(s)$,
the normalization would be by the quantity
\begin{equation}
\label{general}
\sigma = \frac{\sqrt{\sum_{k=1}^n v(S_k) \cdot (W_k)^2}}{\sum_{k=1}^n W_k};
\end{equation}
needless to say, $v(s) = r(s) \cdot (1-r(s))$ for a Bernoulli variate
taking the value $1$ with probability $r(s)$
and the value $0$ with probability $1-r(s)$, consistent with~(\ref{Bernoulli}).
``Normalization'' means considering the ratios $G/\sigma$ and $H/\sigma$
rather than the unnormalized $G$ and $H$
from~(\ref{kolmogorov-smirnov}) and~(\ref{kuiper}).

In many cases in practice, such a large sample of the full population
is available that $r$ and $v$ can be estimated to high accuracy from the data;
\citet{tygert_full} elaborates methods for such estimation. The estimates
of $v$ used in Section~\ref{results} of the present paper adjust for bias
via dividing by 1 minus the ratio of the sum of the squares of the weights
to the square of the sum of the weights.
When all weights are equal, this adjustment simply multiplies by $m / (m - 1)$,
where $m$ is the number of weights,
so that the estimate of the variance involves dividing by $m-1$
(rather than by $m$) in the calculation of the empirical variance.
While the present paper makes no claim whatsoever as to the proper resolution
of the historic debate about whether estimates of the variance
should involve dividing by $m$ or by $m-1$, the estimates (when used
in the context of the cumulative statistics) did improve very slightly
when adjusting for bias in the estimates.

\subsection{Calibration of P-values for the Kolmogorov-Smirnov statistic
and the Kuiper statistic}
\label{calibration}

This subsection derives Corollary~\ref{corollary},
providing a method for the calculation of attained significance levels
(P-values) for the Kolmogorov-Smirnov and Kuiper metrics introduced
in the previous subsection.

Propositions 1--4 of~\citet{diebolt} prove the following theorem.
Technically, \citet{diebolt} provides much stronger and more general results,
characterizing not only convergence but also the convergence rates.
See also closely related results of~\citet{stute}.
Earlier results of~\citet{delgado} motivated the work
of~\citet{diebolt} and~\citet{stute} (among others),
and are also closely related to the metrics of~\citet{tygert_two}.
The proofs of~\citet{delgado} are in some ways simpler and easier to grasp,
despite being restricted to a somewhat more special case,
and are a superb starting point in addition to being of substantial
independent importance, both practically and theoretically.

\begin{theorem}
\label{kscal}
Assume the null hypothesis that the subpopulation has no expected deviation
from the full population (that is, $\E[R_k] = r(S_k)$
for $k = 1$,~$2$, \dots, $n$) and that the third moment obeys
$\E[|R_k - r(S_k)|^3] \le C (v(S_k))^{3/2}$ for $k = 1$,~$2$, \dots, $n$,
where $C$ is a finite positive real number that does not depend on $n$.
Suppose that the scores $S_1$, $S_2$, \dots, $S_n$ are distinct for each $n$
and
$\max_{1 \le k \le n} v(S_k) \cdot (W_k)^2 / \sum_{j=1}^n v(S_j) \cdot (W_j)^2$
converges to 0 in the limit as $n$ becomes large.
Then, with $G$ defined in~(\ref{kolmogorov-smirnov})
and $\sigma$ defined in~(\ref{general}),
the normalized Kolmogorov-Smirnov statistic $G/\sigma$
for measuring deviation of a subpopulation from the full population
converges in distribution to the maximum of the absolute value
of the standard Brownian motion over the unit interval $[0, 1]$.
The normalized Kolmogorov-Smirnov statistic $G/\sigma$
for measuring calibration converges in distribution
to the maximum of the absolute value of the standard Brownian motion
over the unit interval $[0, 1]$, too, when taking the expected response
at each score to be equal to the score, that is, $r(s) = s$
for every score $s$.
\end{theorem}

The theorems of~\citet{diebolt} similarly yield the analogous theorem
for the Kuiper statistic:

\begin{theorem}
\label{kucal}
Assume the null hypothesis that the subpopulation has no expected deviation
from the full population (that is, $\E[R_k] = r(S_k)$
for $k = 1$,~$2$, \dots, $n$) and that the third moment obeys
$\E[|R_k - r(S_k)|^3] \le C (v(S_k))^{3/2}$ for $k = 1$,~$2$, \dots, $n$,
where $C$ is a finite positive real number that does not depend on $n$.
Suppose that the scores $S_1$, $S_2$, \dots, $S_n$ are distinct for each $n$
and
$\max_{1 \le k \le n} v(S_k) \cdot (W_k)^2 / \sum_{j=1}^n v(S_j) \cdot (W_j)^2$
converges to 0 in the limit as $n$ becomes large.
Then, with $H$ defined in~(\ref{kuiper})
and $\sigma$ defined in~(\ref{general}),
the normalized Kuiper statistic $H/\sigma$
for measuring deviation of a subpopulation from the full population
converges in distribution to the range of the standard Brownian motion
over the unit interval $[0, 1]$.
(The range is the maximum minus the minimum.)
The normalized Kuiper statistic $H/\sigma$ for measuring calibration
converges in distribution to the range of the standard Brownian motion
over the unit interval $[0, 1]$, too, when taking the expected response
at each score to be equal to the score, that is, $r(s) = s$
for every score $s$.
\end{theorem}

Putting everything together yields the following.
\begin{corollary}
\label{corollary}
Taking 1 and subtracting the function $D$ from~(\ref{kscdf})
applied to the normalized Kolmogorov-Smirnov statistic $G/\sigma$
yields estimates which converge in distribution to the asymptotic P-value
as $n$ becomes large (due to Theorems~\ref{maincombomax}
and~\ref{kscal}) --- this is $1 - D(G/\sigma)$.
Evaluating 1 minus the function $F$ from~(\ref{cdf})
applied to the normalized Kuiper statistic $H/\sigma$
yields estimates which converge in distribution to the asymptotic P-value
as $n$ becomes large (due to Theorems~\ref{maincomborange}
and~\ref{kucal}) --- this is $1 - F(H/\sigma)$.
The Kolmogorov-Smirnov metric $G$ is defined in~(\ref{kolmogorov-smirnov}),
the Kuiper metric $H$ is defined in~(\ref{kuiper}),
and the normalizing factor $\sigma$ is defined in~(\ref{general}),
with (\ref{general}) reducing to~(\ref{Bernoulli})
when the responses are Bernoulli variates.
\end{corollary}

\subsection{Ties in ranking scores can be treated as weighted samples}
\label{reduction}

Subsection~\ref{tygertsummary} above suggests making minute
random perturbations to the scores in order to ensure that the scores
are distinct from each other.
The present subsection proposes an alternative to breaking ties at random.
The present subsection constructs from the original data
a weighted data set that modifies the scores, weights, and responses
such that the new scores are unique and
(together with the new weights and responses) yield cumulative statistics
that are consistent with those computed with the original data.
This reduces the problem of analyzing data with scores that may not be unique
to the problem of analyzing a weighted data set with scores that are unique
by construction. The earlier subsections have already detailed
how to process weighted samples whose scores are all distinct from each other.

The formulation of the present subsection is merely an alternative,
not necessarily superior.
The alternative formulation requires no randomization of the data analysis,
unlike the earlier analyses.
The graphs of the earlier analyses directly displayed all members
of the original data set, omitting no one.
In contrast, for each score that multiple individuals share,
the graphs for the formulation of the present subsection display
only the average of those multiple individuals' responses.
Nevertheless, the corresponding scalar summary statistics
have the same interpretations and asymptotic calibrations of P-values.
Thus, the earlier and new formulations have advantages and disadvantages
relative to each other (though none of the disadvantages is substantial,
admittedly). Both are good options to have available.

The previous subsections directly analyzed only data sets in which the scores
are all unique:
\begin{equation}
\label{unique}
S_1 < S_2 < \dots < S_n,
\end{equation}
where the inequalities are all strict.
The present subsection considers the case in which each score $S_k$
may appear multiple times --- say $n_k$ times --- in the data set.
With this notation of $n_k$ specifying the degeneracy of score $S_k$,
we define $W_k$ to be the sum of all $n_k$ of the original weights
associated with score $S_k$; denoting the original weights
by $W_k^{(1)}$, $W_k^{(2)}$, \dots, $W_k^{(n_k)}$, we thus define
\begin{equation}
\label{newweight}
W_k = \sum_{j=1}^{n_k} W_k^{(j)}
\end{equation}
for $k = 1$, $2$, \dots, $n$.
We define $R_k$ to be the weighted average
of all $n_k$ of the original real-valued responses associated with score $S_k$;
denoting the original responses
by $R_k^{(1)}$, $R_k^{(2)}$, \dots, $R_k^{(n_k)}$, we thus define
\begin{equation}
R_k = \frac{\sum_{j=1}^{n_k} R_k^{(j)} \, W_k^{(j)}}
           {\sum_{j=1}^{n_k} W_k^{(j)}}
\end{equation}
for $k = 1$, $2$, \dots, $n$.
This yields a data set consisting of the weighted sample
$(S_k, R_k, W_k)$ for $k = 1$, $2$, \dots, $n$,
where $S_k$ is the score, $R_k$ is the associated response,
and $W_k$ is the associated weight.
So this new weighted data set contains $n$ members
$(S_k, R_k, W_k)$ for $k = 1$, $2$, \dots, $n$,
whereas the original data set contains $\sum_{k=1}^n n_k$ members
$(S_k^{(j)}, R_k^{(j)}, W_k^{(j)})$ for $k = 1$, $2$, \dots, $n$;
$j = 1$, $2$, \dots, $n_k$.
Analyzing the new weighted data set via the cumulative statistics
is a good way to analyze the original data set.
And, unlike the scores for the original data set,
the scores for the new weighted data set are guaranteed to be unique.

We now show that the cumulative statistics for the original and new data sets
are consistent with each other.

The cumulative differences for the new data are
\begin{equation}
\label{newcumdiffs}
C_{\ell} = \frac{\sum_{k=1}^{\ell} (R_k - r(S_k)) \, W_k}{\sum_{k=1}^n W_k}
\end{equation}
for $\ell = 1$, $2$, \dots, $n$,
where $r$ is the regression function we seek to test;
when testing calibration, the regression function $r$
is simply the identity function $r(s) = s$ for every real number $s$.
When comparing a subpopulation to the full population,
$r(S_k)$ would be the (weighted) average of responses
from the full population at scores that are closer to $S_k$
than to any other of the scores $S_1$, $S_2$, \dots, $S_n$.
We set $C_0 = 0$, too. 

Let us denote by $v(R_k)$ the variance of the response $R_k$
corresponding to the score $S_k$ under the null hypothesis,
where the null hypothesis makes assumptions about the original data directly
(so that inferences about $R_k$ take into account the fact that $R_k$
is a weighted average of other random variables, instead of considering $R_k$
to be a single response variable).
For example, under the null hypothesis of perfect calibration
with each response drawn independently from a Bernoulli distribution,
\begin{equation}
\label{Bernoullivar}
v(R_k) = S_k \, (1 - S_k) \, \frac{\sum_{j=1}^{n_k} \left(W_k^{(j)}\right)^2}
                                  {\left(\sum_{j=1}^{n_k} W_k^{(j)}\right)^2},
\end{equation}
since $S_k \, (1 - S_k)$ is the variance
of the Bernoulli distribution whose expected value is $r(S_k) = S_k$.
Calibration need not be the only hypothesis of interest to test.
Under the null hypothesis that a subpopulation being assessed does not deviate
from the function $r$ for the full population,
an estimate of $v(R_k)$ can be the (weighted) average of variances of responses
from the full population at scores that are closer to $S_k$
than to any other of the scores $S_1$, $S_2$, \dots, $S_n$
(assuming as always that the responses are all independent),
multiplied by the same factor from~(\ref{Bernoullivar}), namely
\begin{equation}
\label{factor}
\frac{\sum_{j=1}^{n_k} \left(W_k^{(j)}\right)^2}
     {\left(\sum_{j=1}^{n_k} W_k^{(j)}\right)^2};
\end{equation}
indeed, the independence of all the responses yields that $v(R_k)$
is equal to the quantity in~(\ref{factor}) times the variance of $R_k^{(j)}$,
for every $j = 1$, $2$, \dots, $n_k$; $k = 1$, $2$, \dots, $n$.
\cite{tygert_full} gives the details.
Since we assumed that the responses are independent, the variance
of $C_{\ell}$ from~(\ref{newcumdiffs}) under the null hypothesis is
\begin{equation}
\label{total}
(\sigma_{\ell})^2 = \frac{\sum_{k=1}^{\ell} v(R_k) \cdot (W_k)^2}
                         {(\sum_{k=1}^n W_k)^2}
\end{equation}
for $\ell = 1$, $2$, \dots, $n$.

We also consider similar cumulative differences for the original data set
in which the scores are perturbed infinitesimally at random
(so that the scores become unique):
\begin{equation}
\label{cumdiffs}
B_{\ell} = \frac{\sum_{k=1}^{\ell} \sum_{j=1}^{n_k}
                 \left(R_k^{(j)} - r(S_k)\right) \, W_k^{(j)}}
                {\sum_{k=1}^n \sum_{j=1}^{n_k} W_k^{(j)}}
         = \frac{\sum_{k=1}^{\ell} (R_k - r(S_k)) \sum_{j=1}^{n_k} W_k^{(j)}}
                {\sum_{k=1}^n \sum_{j=1}^{n_k} W_k^{(j)}}
\end{equation}
for $\ell = 1$, $2$, \dots, $n$,
where the ordering of $R_k^{(1)}$, $R_k^{(2)}$, \dots, $R_k^{(n_k)}$
(and the corresponding weights) is randomized
for each $k = 1$, $2$, \dots, $n$. We set $B_0 = 0$, too.

We define abscissae via the aggregations
\begin{equation}
A_{\ell} = \frac{\sum_{k=1}^{\ell} \sum_{j=1}^{n_k} W_k^{(j)}}
                {\sum_{k=1}^n \sum_{j=1}^{n_k} W_k^{(j)}}
         = \frac{\sum_{k=1}^{\ell} W_k}{\sum_{k=1}^n W_k}
\end{equation}
for $\ell = 1$, $2$, \dots, $n$, where the latter equality
follows from~(\ref{newweight}). We set $A_0 = 0$, too.
Combining~(\ref{newweight}), (\ref{newcumdiffs}), and~(\ref{cumdiffs})
shows that $B_{\ell} = C_{\ell}$
for all $\ell = 1$, $2$, \dots, $n$.
Therefore, the piecewise linear graph connecting
the points $(A_{\ell},\,B_{\ell} / \sigma_n)$
for $\ell = 0$, $1$, $2$, \dots, $n$
and the piecewise linear graph connecting
the points $(A_{\ell},\,C_{\ell} / \sigma_n)$
for $\ell = 0$, $1$, $2$, \dots, $n$ are the same.
This demonstrates that the cumulative statistics for the original
and new data sets are consistent with each other.
Indeed, the corresponding graph of cumulative differences
for the original data with its scores perturbed very slightly
(so that the scores become unique) is the same
aside from the other graphs linearly interpolating
from each score $S_k$ to the next greatest score, $S_{k+1}$,
rather than interpolating linearly from each and every perturbed score
to the next greatest perturbed score.

Theorems~\ref{kscal} and~\ref{kucal} and their Corollary~\ref{corollary}
yield the same consequences for all cumulative statistics considered here,
under the condition (expressed in the notation of the present subsection):
\begin{equation}
\label{complicated}
\frac{\max_{1 \le k \le n}
      \left[ v(R_k) \cdot \frac{\left(\sum_{j=1}^{n_k} W_k^{(j)}\right)^2}
                               {\sum_{j=1}^{n_k} \left(W_k^{(j)}\right)^2}
      \cdot \max_{1 \le j \le n_k} \left(W_k^{(j)}\right)^2 \right]}
     {\sum_{k=1}^n
      \left[ v(R_k) \cdot \left(\sum_{j=1}^{n_k} W_k^{(j)}\right)^2 \right]}
\end{equation}
converges to 0 in the limit as $n$ becomes large.
The denominator in~(\ref{complicated}) is simply the numerator
in~(\ref{complicated}) after replacing the maximizations with sums.

To summarize: the cumulative statistics for the original data set
can require perturbing the scores slightly in order to break degeneracies,
unlike the cumulative statistics for the new weighted data.
The randomization does preserve more information about the original data,
as the associated graph of cumulative differences displays
the response of every single individual from the original data set.
The new weighted data set instead avoids any randomization
but, for each score that multiple members share,
averages together the multiple members' responses.
Thus both the previous approaches and that of the present subsection
have pros and cons relative to each other.
That said, the approaches are more similar than different;
neither has any substantial drawback.

\section{Results}
\label{results}

The present section illustrates (via examples and plots)
the numerical and graphical methods of the preceding section.\footnote{Software
in Python 3 that automatically reproduces all results and figures reported
in the present section is available
at \url{https://github.com/facebookresearch/cdeets}}
Subsection~\ref{validation} verifies the methods numerically,
double-checking the rigorous proofs given earlier.
Subsection~\ref{analysis} applies the methods to a popular data set
from the U.S.\ Census Bureau.

\subsection{Numerical validation}
\label{validation}

This subsection presents numerical verifications of the methods
of the preceding section. The numerical validation is purely supplemental,
as the proofs given earlier are complete on their own. The numerical results
are nice and concrete, possibly easier to digest than the detailed proofs.

Figures~\ref{kuiper_plot} and~\ref{kolmogorov-smirnov_plot}
plot $1 - F(x)$ versus $x$ and $1 - D(x)$ versus $x$,
respectively, where $F$ is defined in~(\ref{cdf})
and $D$ is defined in~(\ref{kscdf}).
The calculation for $F$ truncates the series in~(\ref{cdf}) after $n$ terms,
where $n = n(x)$ is the least integer such that~(\ref{kuterms})
guarantees full double-precision accuracy (with $\epsilon \approx$ 2.2E--16).
Similarly, the calculation for $D$ truncates the series in~(\ref{kscdf})
after $n$ terms, where $n = n(x)$ is the least integer
such that~(\ref{ksterms}) guarantees full double-precision accuracy
(again with $\epsilon \approx$ 2.2E--16).
To give an indication of how another sub-Gaussian distribution decays,
Figure~\ref{gaussian_plot} plots $1 - \Phi(x)$ versus $x$,
where $\Phi$ is the cumulative distribution function
for the standard normal distribution.

Formula~1.4 of~\citet{feller} and Formula~46 of~\citet{masoliver} give
the means of the distributions associated
with the cumulative distribution functions
$F$ and $D$ defined in~(\ref{cdf}) and~(\ref{kscdf}),
as $2 \sqrt{2/\pi} \approx 1.5958$
and $\sqrt{\pi/2} \approx 1.2533$, respectively.
The horizontal positions of the vertical dotted lines labeled ``mean''
in Figures~\ref{kuiper_plot} and~\ref{kolmogorov-smirnov_plot}
are at these mean values.
A unit test of the implementations of the cumulative distribution functions
is to numerically evaluate the means.
Using a Gauss-Chebyshev quadrature of order 100,000 to integrate
$1 - F(x)$ and $1 - D(x)$ from $x =$ 1E--8 to $x = 8$
yields the correct means to better than 8-digit relative accuracy
in the implemented codes, thus passing this unit test.

Figures~\ref{kuc} and~\ref{ksc} plot the calibration curves
for the Kuiper and Kolmogorov-Smirnov statistics, respectively.
The calibration curves are the empirical cumulative distribution functions
of the asymptotic P-values for calibration calculated
for 100,000 data sets generated
by drawing independent Bernoulli responses at the scores,
with the probability of success in the Bernoulli distribution
being exactly equal to the score (so that the data is perfectly calibrated,
by construction). Perfectly calibrated P-values would follow
the uniform distribution over the unit interval $[0, 1]$
under the null hypothesis, and so ideally the plotted
empirical cumulative distribution functions should approach
the cumulative distribution function for the uniform distribution
as the sample size increases. The cumulative distribution function
for the uniform distribution over the unit interval $[0, 1]$
is the line connecting the origin $(0, 0)$ to the point $(1, 1)$;
each plot displays a dashed line to indicate the ideal calibration curve.
The other curves are the empirical cumulative distribution functions
of the P-values for data sets with sample sizes $n =$ 100, 1,000, 10,000;
as expected, the curve closest to the diagonal dashed line in each plot 
is that for $n =$ 10,000, the next closest is for $n =$ 1,000,
and the farthest is for $n =$ 100.
The weights in these synthetically generated data sets are uniform (all equal),
just for simplicity.

Figures~\ref{kuc} and~\ref{ksc} illustrate Corollary~\ref{corollary},
with convergence to the ideal calibration
that~\citet{delgado}, \citet{diebolt}, and~\citet{stute}
prove as the scores become dense in the unit interval $[0, 1]$
(the scores are quite dense already with $n =$ 10,000, for example).
Notice that the empirical curves all lie entirely below
the diagonal dashed line, in accordance
with the calculated finite-sample P-values being conservatively calibrated
(the P-values are not smaller on average than expected).

The ends of the captions of Figures~\ref{la}, \ref{stanislaus}, and~\ref{napa}
from the following subsection report P-values evaluated
using Corollary~\ref{corollary}.
Attained significance levels (P-values) for all methods of~\citet{tygert_full}
can also be calibrated and calculated directly using Corollary~\ref{corollary},
under the assumption that, for each of the scores from the subpopulation,
the full population contains many members whose scores are closer to the score
from the subpopulation than to other scores from the subpopulation;
if this assumption is invalid, then the statistics fed
into the cumulative distribution functions require adjustment
to account for the additional stochasticity,
as described by~\citet{tygert_full} and~\citet{tygert_two}.

\subsection{Data analysis}
\label{analysis}

This subsection illustrates the methods of the preceding section
by applying the methods to the microdata of the U.S.\ Census Bureau's
2019 American Community Survey.\footnote{The microdata
from the American Community Survey is available for download
via the FTP servers and other means detailed
at \url{https://www.census.gov/programs-surveys/acs/microdata.html}}
We discard every member of the data set for which the weight
(``WGTP'' in the microdata) for the weighted sampling is zero,
as well as every member for which household personal income (``HINCP'') is zero
and every member for which the adjustment factor to income (``ADJINC'')
is reported as missing.
The scores are the logarithm to base 10
of the adjusted household personal income
(the adjusted income is ``HINCP'' times ``ADJINC,'' divided by one million;
the one million accounts for the omission of any decimal point in ``ADJINC''
--- the microdata is integer-valued).
The responses are the variables from the data set specified
in the captions to the figures for this subsection,
namely Figures~\ref{la}--\ref{napa}
(different figures analyze different response variables).
The full population in the survey consists of 134,094 households,
a weighted sample of California.
The subpopulation being compared to the full population
consists of the households in the county specified in the caption
to the corresponding figure.

\section{Discussion and conclusion}
\label{conclusion}

As shown above, the combination of~\citet{feller}, \citet{darling-siegert},
\citet{delgado}, \citet{diebolt}, \citet{stute}, and others trivially yields
computationally efficient and convenient calibration of P-values
for the metrics of~\citet{tygert_full}, metrics very similar
to those of~\citet{kolmogorov} and~\citet{smirnov} and of~\citet{kuiper}
(whose work directly stimulated all the others', including that
of the author of the present paper).
The results of~\citet{delgado}, \citet{diebolt}, and~\citet{stute}
reduce the problem of calibration to the calculation
of the distributions of the range and of the maximum absolute value
of the standard Brownian motion over the unit interval $[0, 1]$;
the results of~\citet{feller} and~\citet{darling-siegert}
completely characterize those distributions.
Simple, straightforward manipulation of the resulting formulae
then yields the cumulative distribution functions required
for calibrating P-values, as detailed in Section~\ref{methods} above.
Section~\ref{methods} also presents two different approaches for processing
data sets in which the scores are not all distinct from each other.
In all cases, implementation is easy; Section~\ref{results} validates
the numerical methods and implementation via plots
of the cumulative distribution functions of the metrics
and of the associated P-values, as well as via checks
against analytic, closed-form expressions, illustrating use of the codes
both on their own and as applied to both real and synthetic data sets.
The software is ready for widespread use
under its permissive MIT copyright license.

\newlength{\imsize}
\setlength{\imsize}{.48\textwidth}
\newlength{\imsized}
\setlength{\imsized}{.59\textwidth}

\begin{figure}
\begin{center}
\parbox{\imsize}{\includegraphics[width=\imsize]{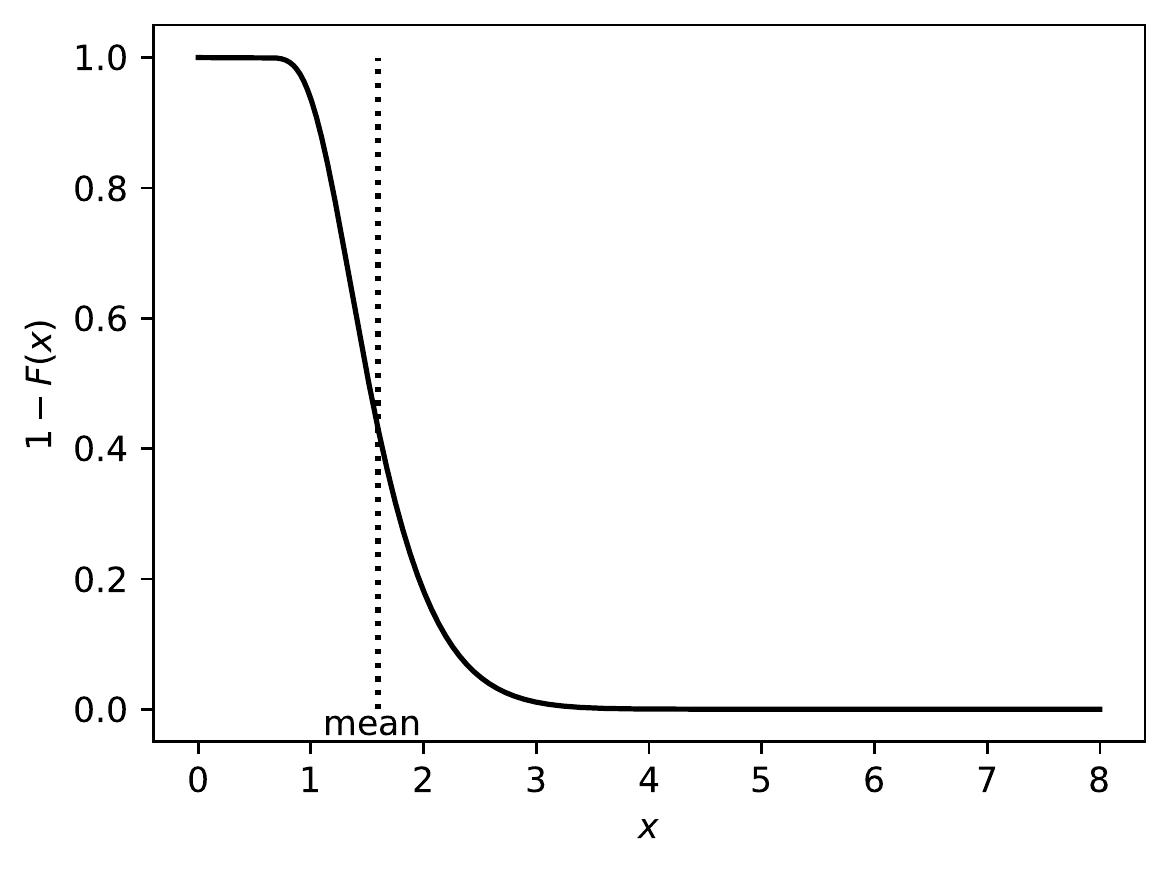}}
\quad
\parbox{\imsize}{\includegraphics[width=\imsize]{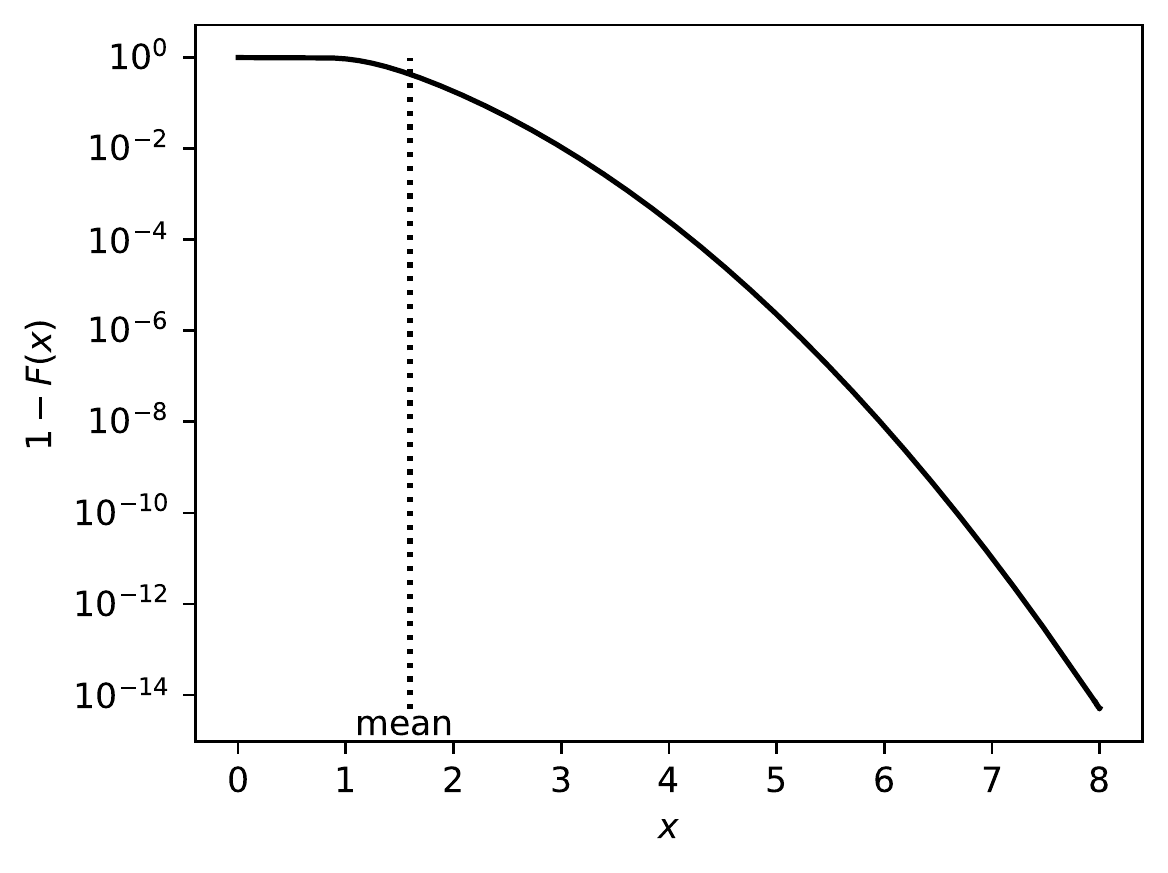}}
\end{center}
\vspace{-.125in}
\caption{Both plots graph $1 - F(x)$ versus $x$, where
$F$ is defined in~(\ref{cdf}) and is central to Corollary~\ref{corollary}.
The plot on the right uses a logarithmic scale for the vertical axis,
unlike the plot on the left. The vertical dotted line indicates
the value of $x$ corresponding to the mean of the distribution
for which $F$ is the cumulative distribution function.}
\label{kuiper_plot}
\end{figure}

\begin{figure}
\begin{center}
\parbox{\imsize}{\includegraphics[width=\imsize]
                {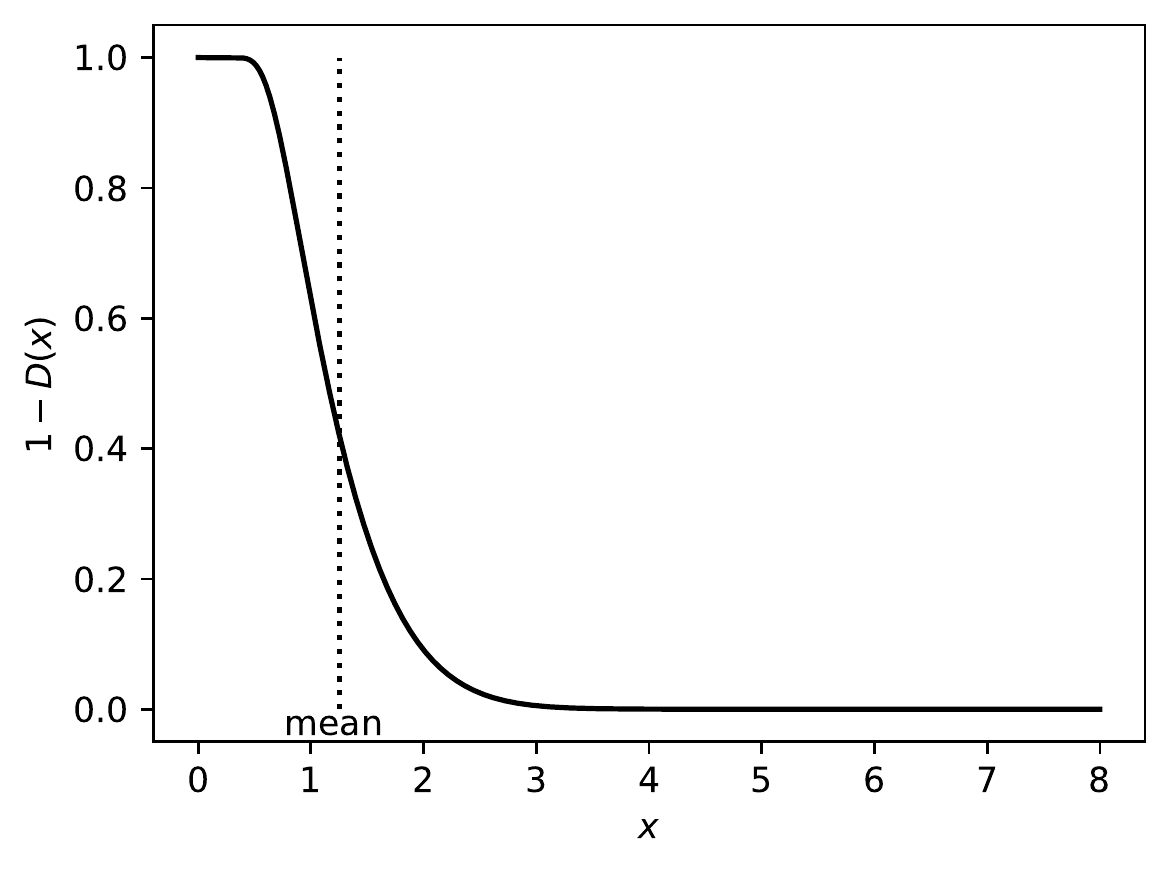}}
\quad
\parbox{\imsize}{\includegraphics[width=\imsize]
                {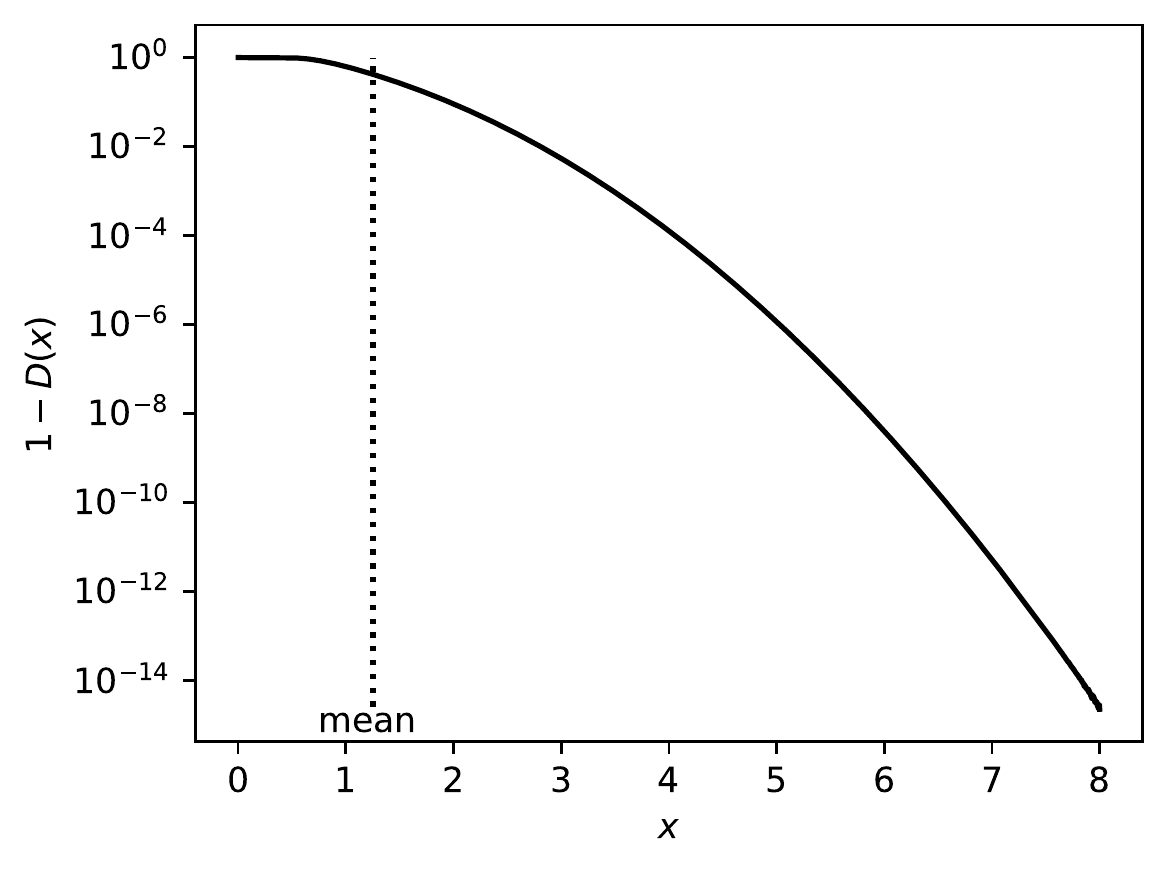}}
\end{center}
\vspace{-.125in}
\caption{Both plots graph $1 - D(x)$ versus $x$, where
$D$ is defined in~(\ref{kscdf}) and is central to Corollary~\ref{corollary}.
The plot on the right uses a logarithmic scale for the vertical axis,
unlike the plot on the left. The vertical dotted line indicates
the value of $x$ corresponding to the mean of the distribution
for which $D$ is the cumulative distribution function.}
\label{kolmogorov-smirnov_plot}
\end{figure}

\begin{figure}
\begin{center}
\parbox{\imsize}{\includegraphics[width=\imsize]{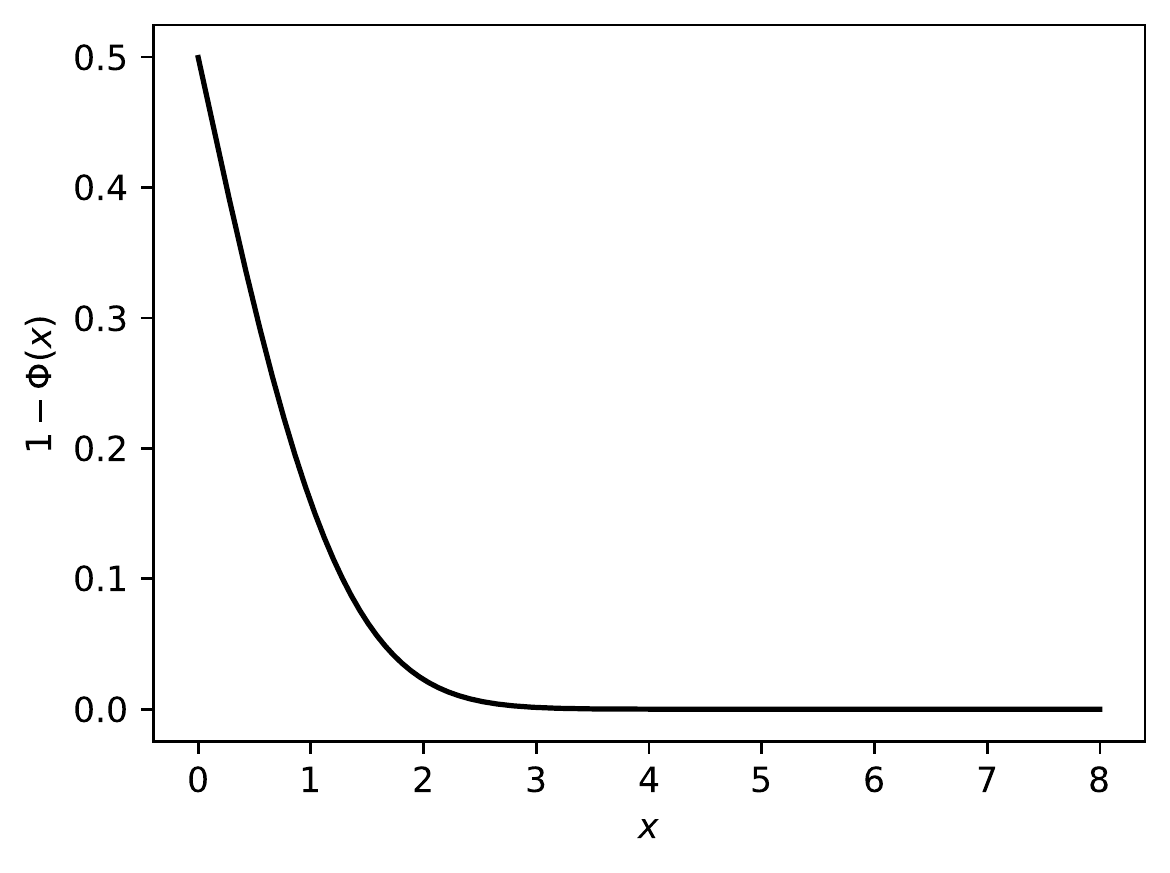}}
\quad
\parbox{\imsize}{\includegraphics[width=\imsize]{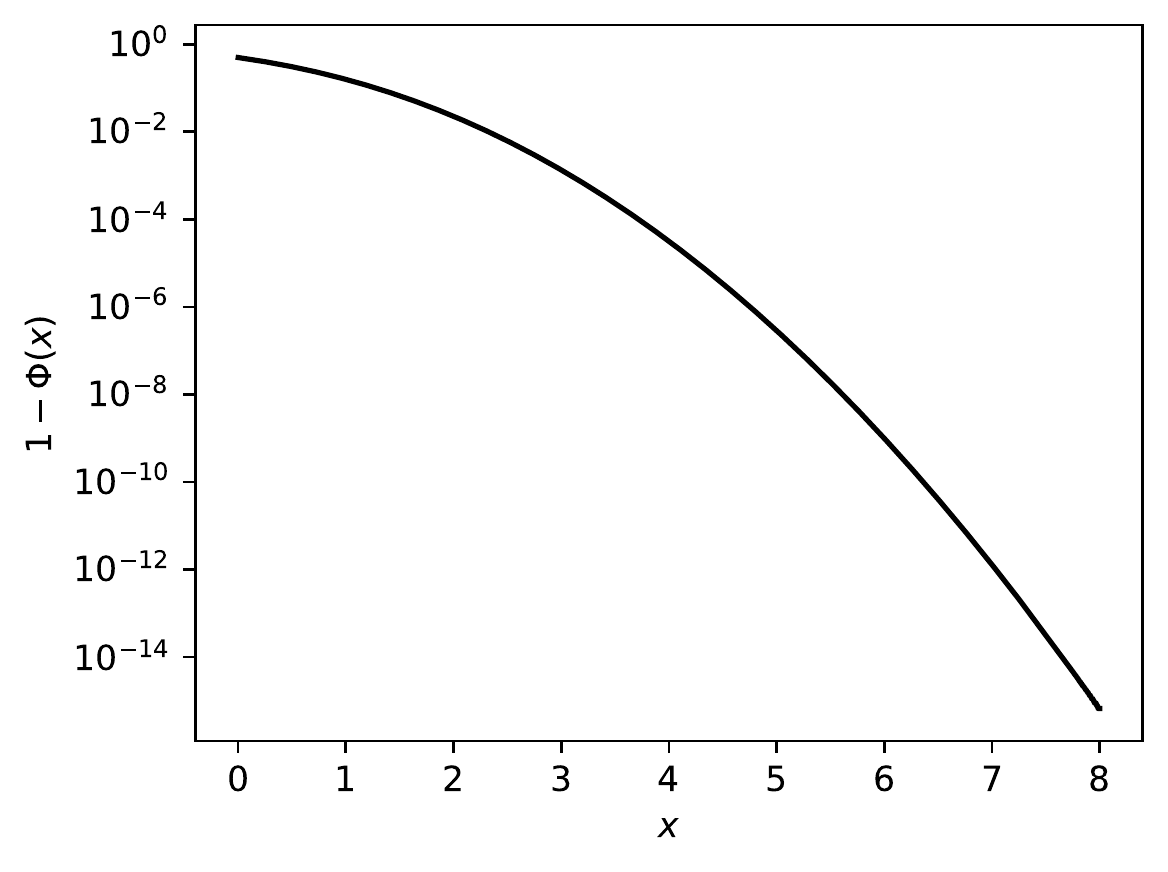}}
\end{center}
\vspace{-.125in}
\caption{Both plots graph $1 - \Phi(x)$ versus $x$,
where $\Phi$ is the cumulative distribution function
for the standard normal distribution;
$\Phi(x) = \int_{-\infty}^x \exp(-y^2/2) \, dy \, / \, \sqrt{2\pi}$.
The plot on the right uses a logarithmic scale for the vertical axis,
unlike the plot on the left.}
\label{gaussian_plot}
\end{figure}

\begin{figure}
\begin{center}
\quad \quad (a)

\parbox{\imsize}{\includegraphics[width=\imsize]
       {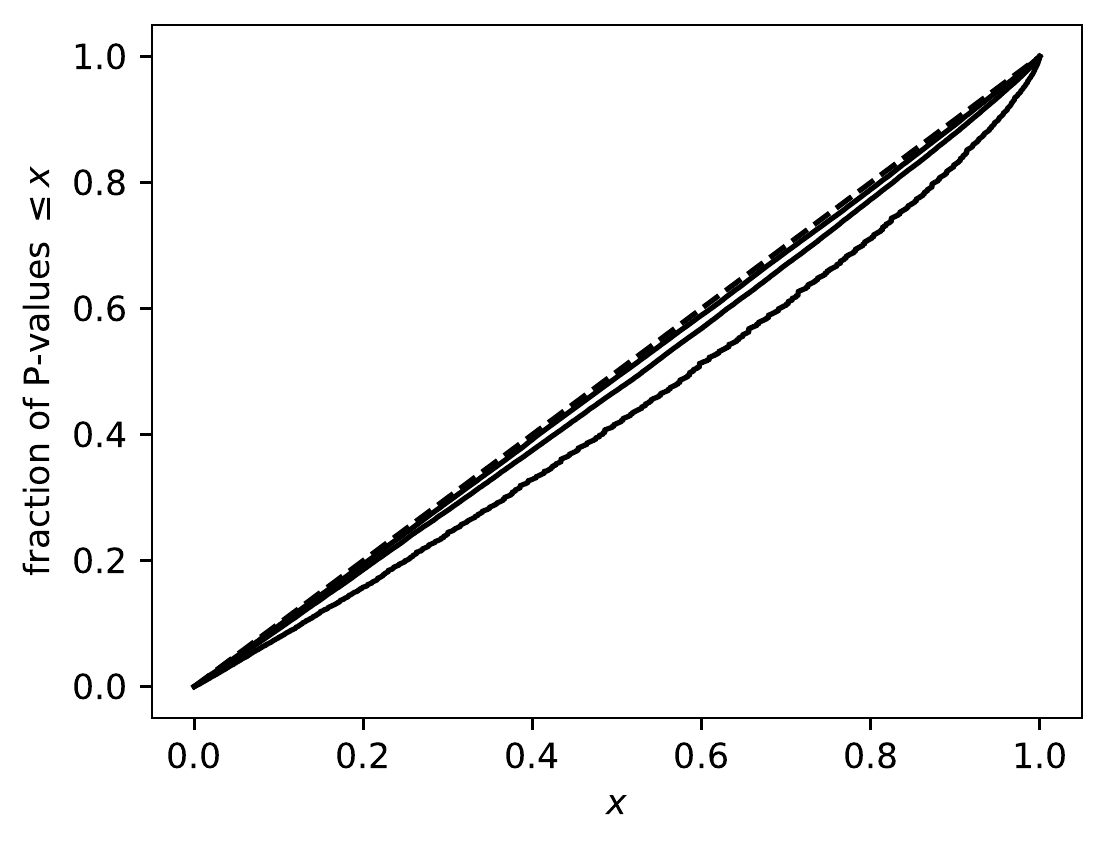}}

\ \quad \quad \hfil (b) \hfil \hfil \hfil \hfil \quad \quad \quad (c)
\hfil \hfil

\parbox{\imsize}{\includegraphics[width=\imsize]
       {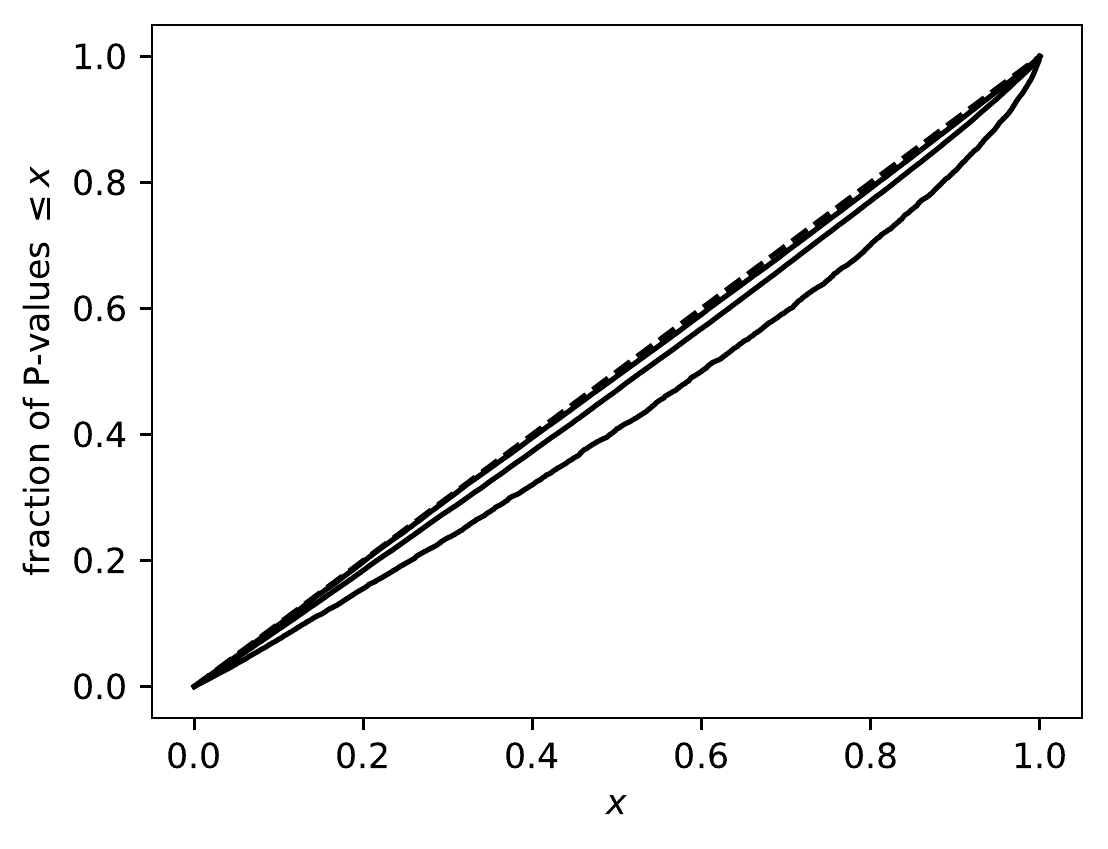}}
\quad
\parbox{\imsize}{\includegraphics[width=\imsize]
       {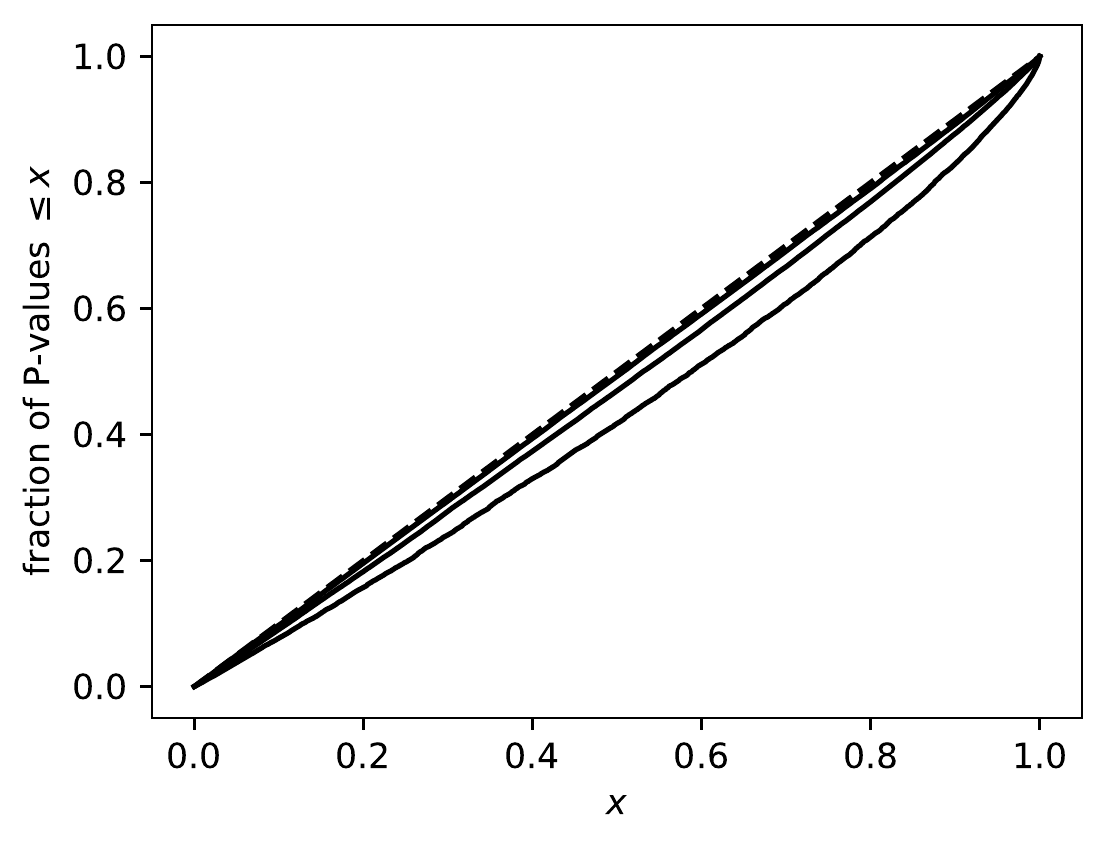}}
\end{center}
\vspace{-.125in}
\caption{Calibration curves (empirical cumulative distribution functions
under the null hypothesis of perfectly calibrated data) of the Kuiper P-value
for calibration for sample sizes $n =$ 100, 1,000, 10,000; in each plot,
the dashed line connects the origin $(0, 0)$ to the point $(1, 1)$
and illustrates perfect calibration, while the curve for $n =$ 10,000
is closest to perfect, $n =$ 1,000 is next closest, and $n =$ 100
is the farthest. Subfigure (a) uses scores equispaced
on the unit interval $[0, 1]$,
(b) squares each of the initially equispaced scores,
and (c) takes the square root of each of the initially equispaced scores.
The score $s$ is the predicted probability,
with the expected response $r(s) = s$ to assess calibration.
Each empirical cumulative distribution function plotted arises
from 100,000 data sets generated independently
while assuming the null hypothesis.}
\label{kuc}
\end{figure}

\begin{figure}
\begin{center}
\quad \quad (a)

\parbox{\imsize}{\includegraphics[width=\imsize]
       {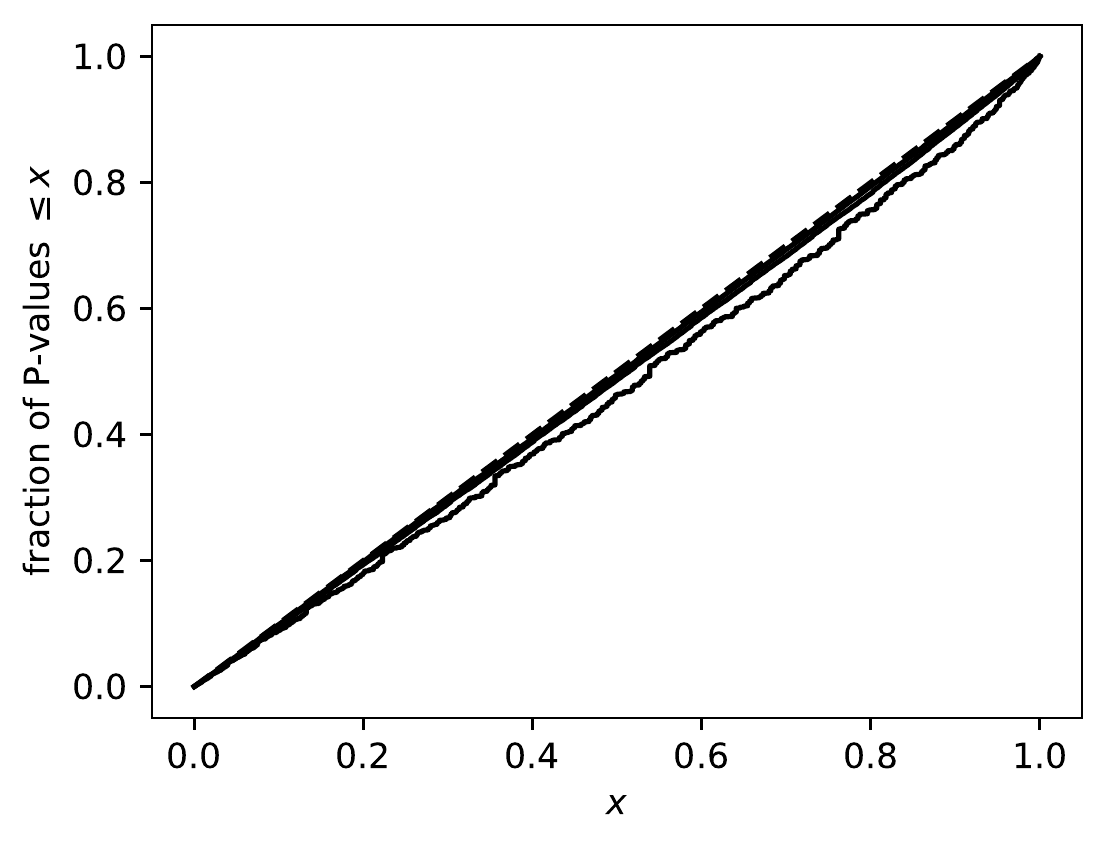}}

\ \quad \quad \hfil (b) \hfil \hfil \hfil \hfil \quad \quad \quad (c)
\hfil \hfil

\parbox{\imsize}{\includegraphics[width=\imsize]
       {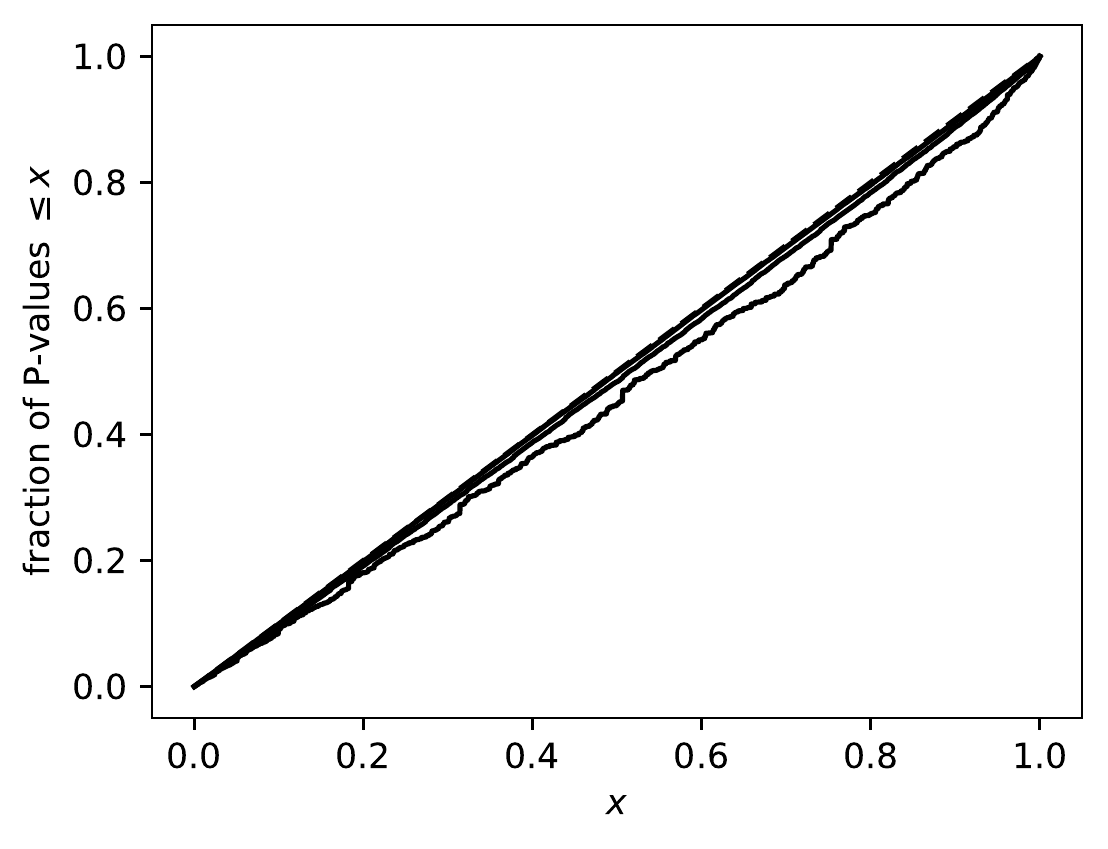}}
\quad
\parbox{\imsize}{\includegraphics[width=\imsize]
       {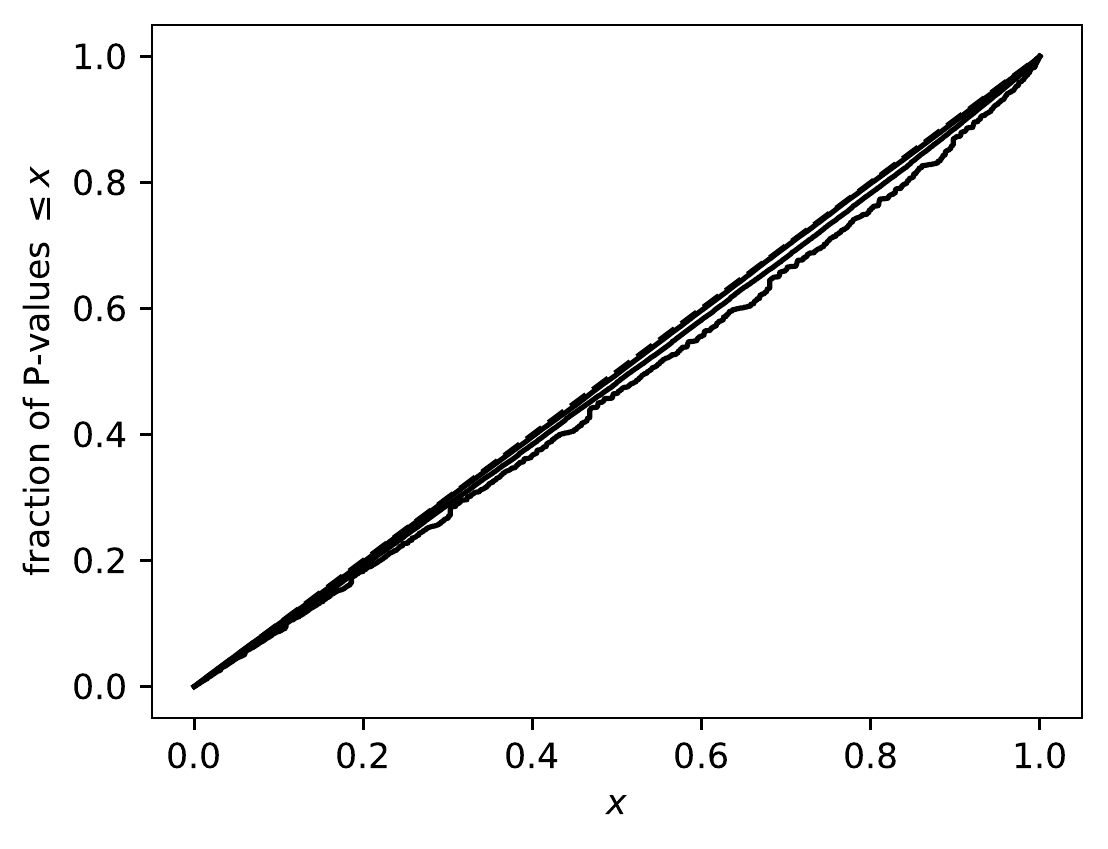}}
\end{center}
\vspace{-.125in}
\caption{Calibration curves (empirical cumulative distribution functions
under the null hypothesis of perfectly calibrated data)
of the Kolmogorov-Smirnov P-value for calibration
for sample sizes $n =$ 100, 1,000, 10,000; in each plot, the dashed line
connects the origin $(0, 0)$ to the point $(1, 1)$ and illustrates
perfect calibration, while the curve for $n =$ 10,000 is closest to perfect,
$n =$ 1,000 is next closest, and $n =$ 100 is the farthest.
Subfigure (a) uses scores equispaced on the unit interval $[0, 1]$,
(b) squares each of the initially equispaced scores,
and (c) takes the square root of each of the initially equispaced scores.
The score $s$ is the predicted probability,
with the expected response $r(s) = s$ to assess calibration.
Each empirical cumulative distribution function plotted arises
from 100,000 data sets generated independently
while assuming the null hypothesis.}
\label{ksc}
\end{figure}

\begin{figure}
\begin{center}
\includegraphics[width=\imsized]
{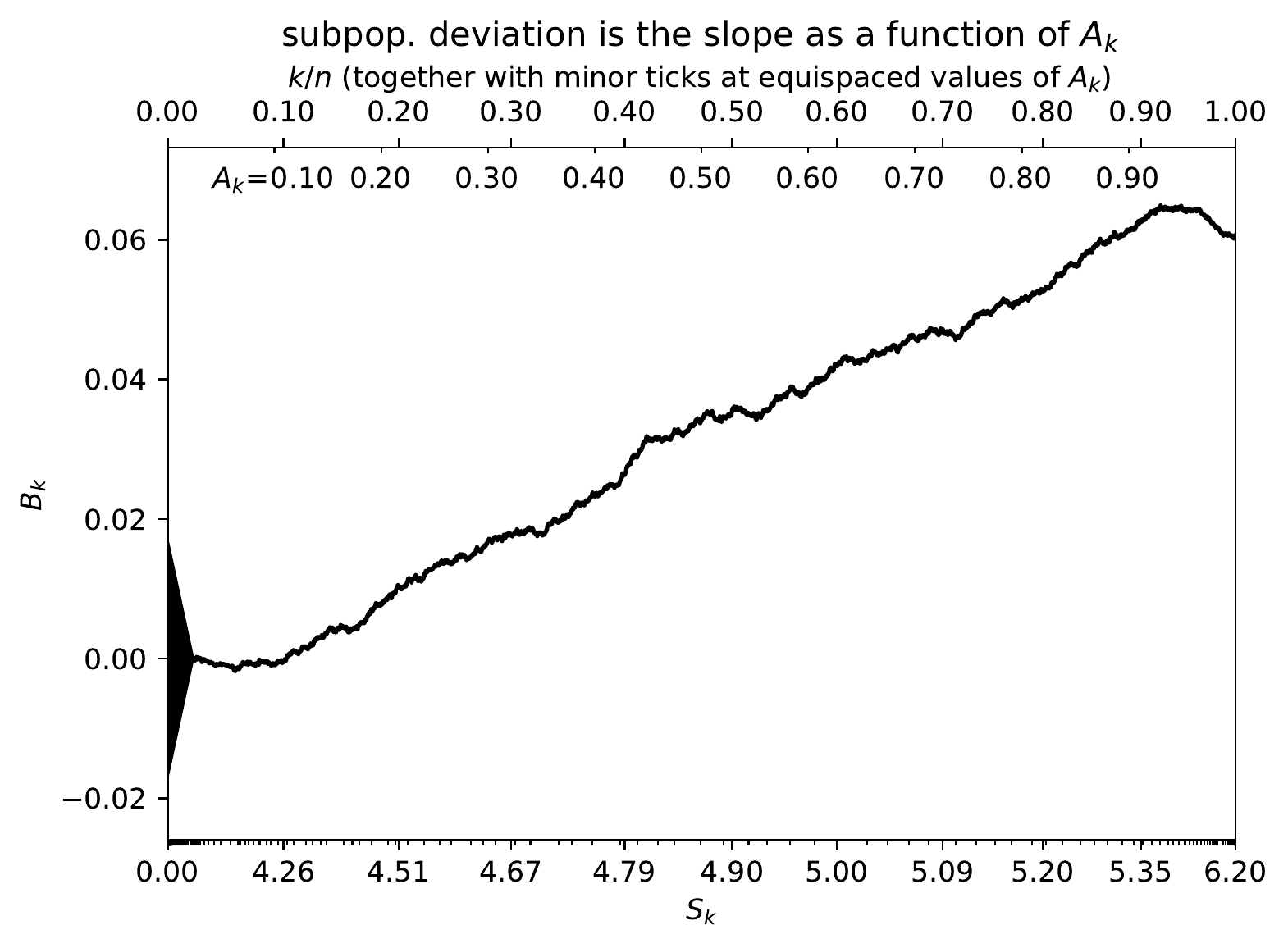}

\vspace{-.5em}

\ \ ($n =$ 35,364)

\

\includegraphics[width=\imsized]
{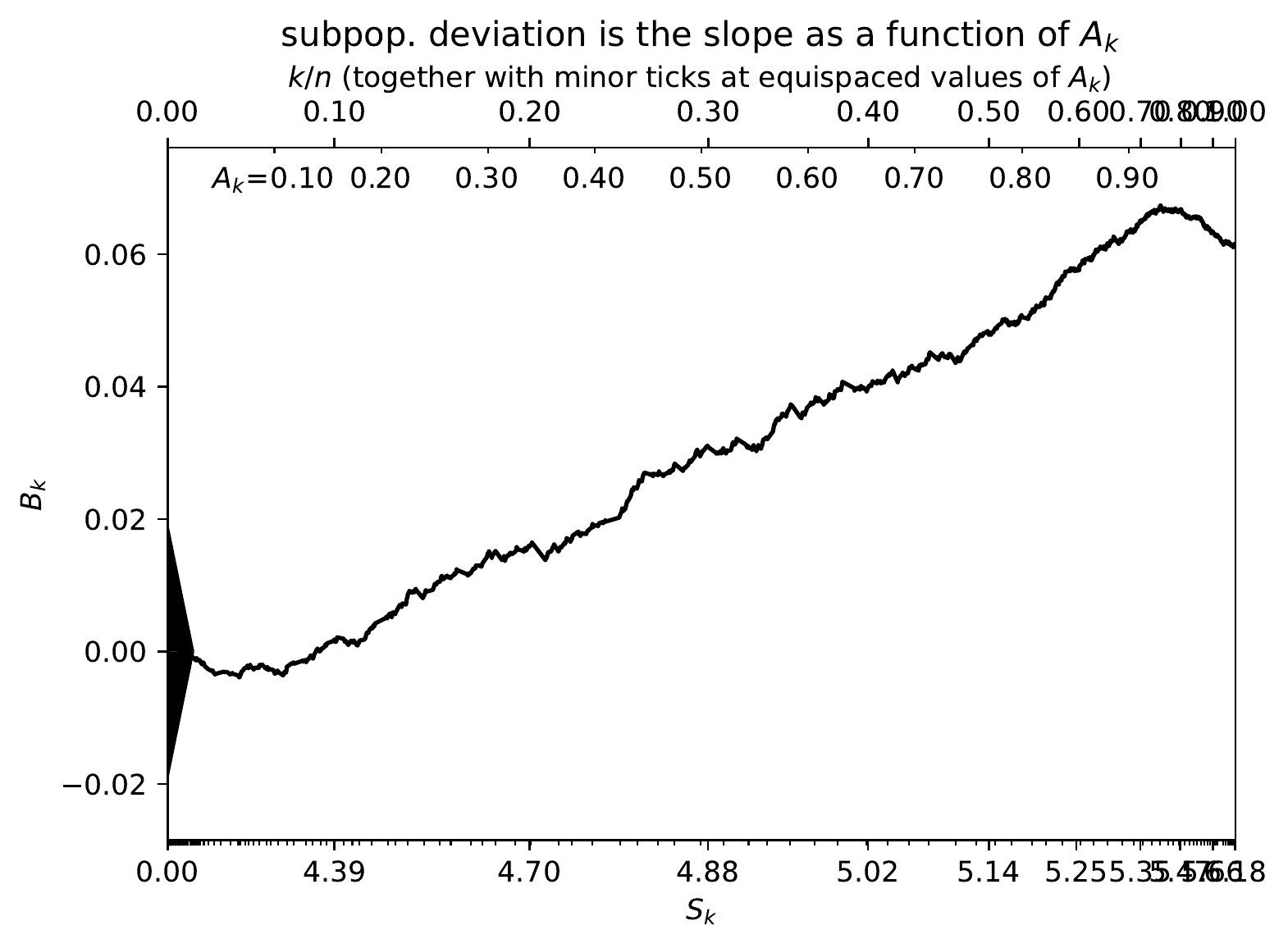}

\vspace{-.5em}

\ \ \ \ \ ($n =$ 5,587)
\end{center}
\vspace{-.125in}
\caption{
Difference in the number of people in a household
between the county of Los Angeles and the entire state of California
(the county is the subpopulation, while the state is the full population).
The scores indicated along the lower horizontal axis are $\log_{10}$
of the adjusted household income, randomly perturbed in the upper plot
by about one part in a hundred million to ensure their uniqueness.
There are 35,364 households representing Los Angeles.
When the scores are perturbed at random ($n =$ 35,364),
Kuiper's statistic $H = 0.06674$, while $H/\sigma = 7.521$;
Kolmogorov's and Smirnov's $G = 0.06495$, while $G/\sigma = 7.319$.
When the responses are averaged for the same score
as in Subsection~\ref{reduction} and displayed in the lower plot ($n =$ 5,587),
Kuiper's statistic $H = 0.07126$, while $H/\sigma_n = 7.213$;
Kolmogorov's and Smirnov's $G = 0.06736$, while $G/\sigma_n = 6.818$.
The P-values for both statistics are 0 to the precision of computations.
These P-values reflect the observed difference of many standard deviations
beyond the expected means.
Deviation of the subpopulation's response (the number of people)
from the full population's is the slope as displayed.}
\label{la}
\end{figure}

\begin{figure}
\begin{center}
\includegraphics[width=\imsized]
{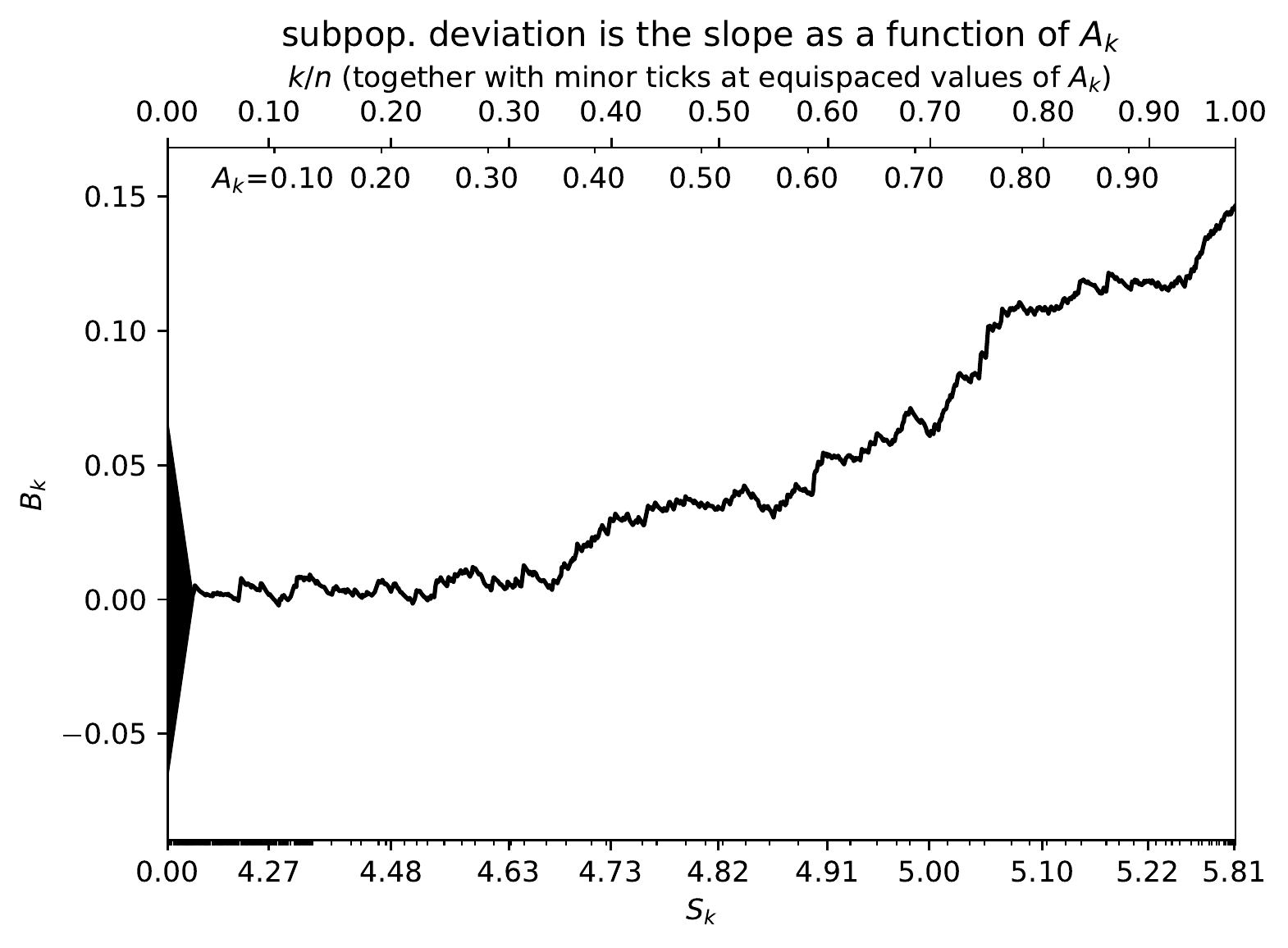}

\vspace{-.5em}

\ \ ($n =$ 1,624)

\

\includegraphics[width=\imsized]
{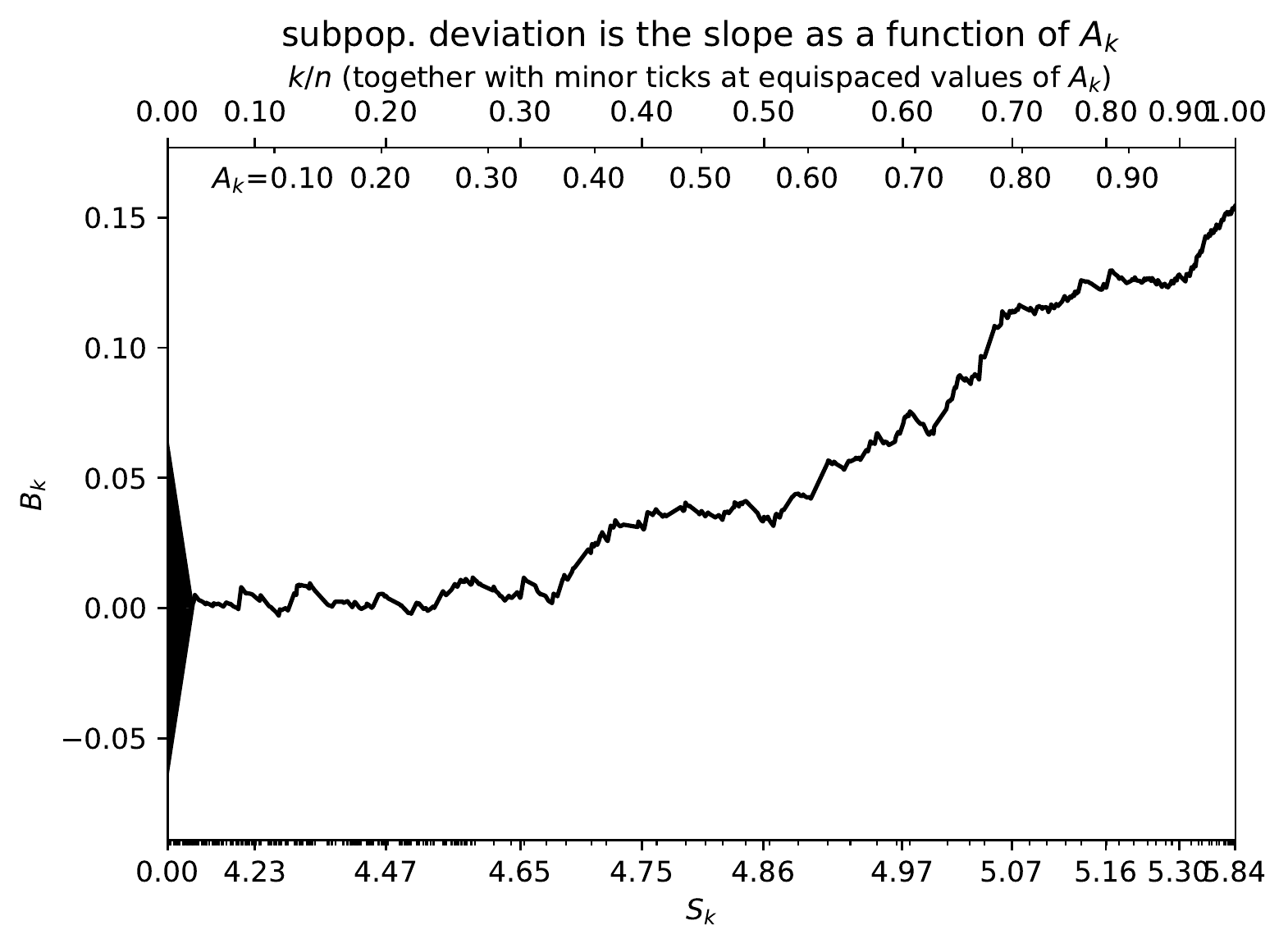}

\vspace{-.5em}

\ \ \ \ \ ($n =$ 892)
\end{center}
\vspace{-.125in}
\caption{
Difference in the number of children related to the head of household
between the county of Stanislaus and the entire state of California
(the county is the subpopulation, while the state is the full population).
The scores indicated along the lower horizontal axis are $\log_{10}$
of the adjusted household income, randomly perturbed in the upper plot
by about one part in a hundred million to guarantee their uniqueness.
There are 1,624 households representing Stanislaus.
When the scores are perturbed at random ($n =$ 1,624),
Kuiper's statistic $H = 0.1489$, while $H/\sigma = 4.373$;
Kolmogorov's and Smirnov's $G = 0.1467$, while $G/\sigma = 4.307$.
When the responses are averaged for the same score ($n =$ 892),
Kuiper's statistic $H = 0.1575$, while $H/\sigma_n = 4.710$;
Kolmogorov's and Smirnov's $G = 0.1547$, while $G/\sigma_n = 4.624$.
The estimates of P-values for Kuiper's statistic are 4.902E--5
and \hbox{0.991E--5}; the estimates of P-values for Kolmogorov's and Smirnov's
are 3.310E--5 and 0.753E--5.
These P-values reflect the observed difference of several standard deviations
beyond the expected means.
Deviation of the subpopulation from the full population is the slope.}
\label{stanislaus}
\end{figure}

\begin{figure}
\begin{center}
\includegraphics[width=\imsized]
{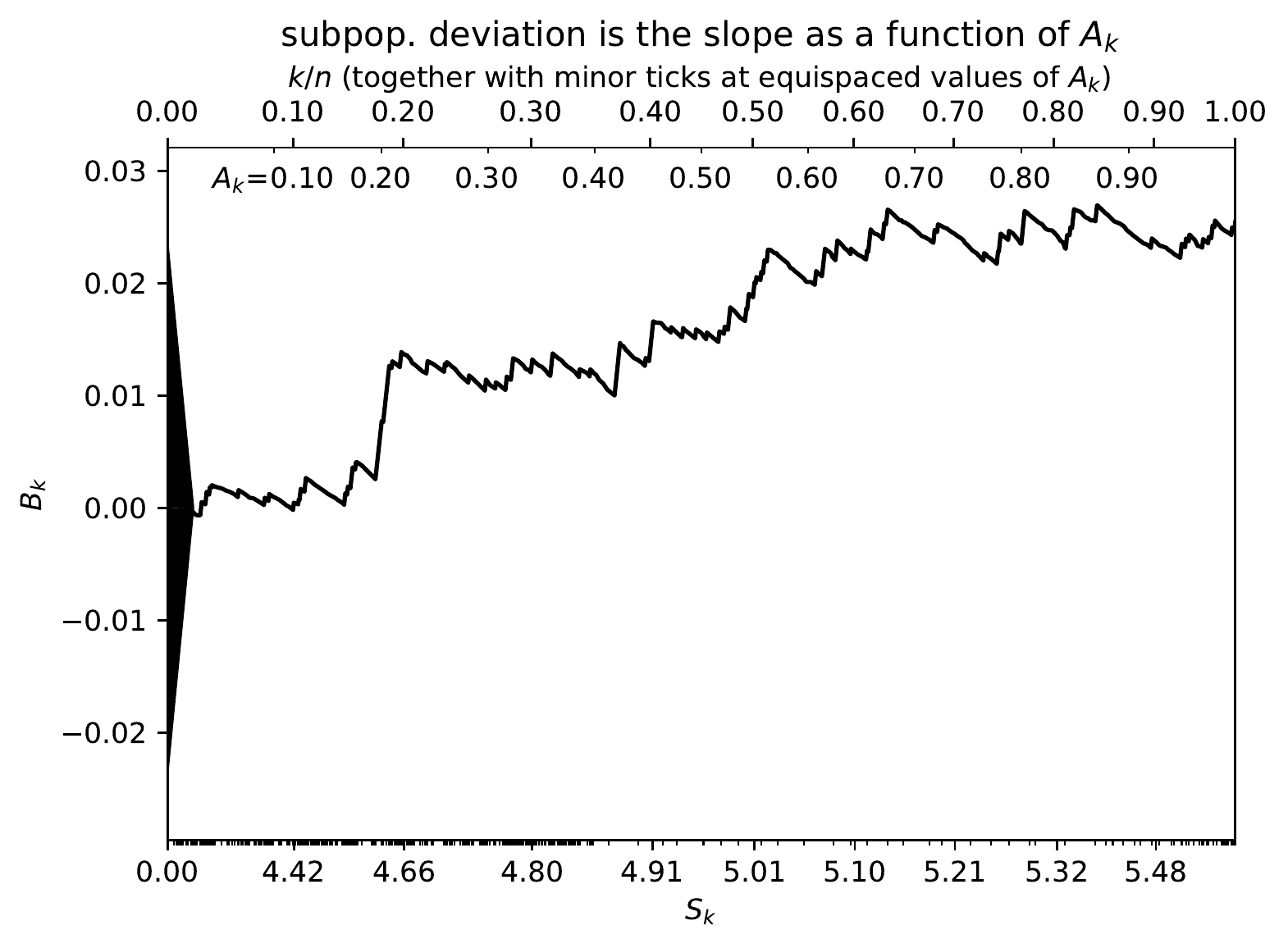}

\vspace{-.5em}

\ \ ($n =$ 679)

\

\includegraphics[width=\imsized]
{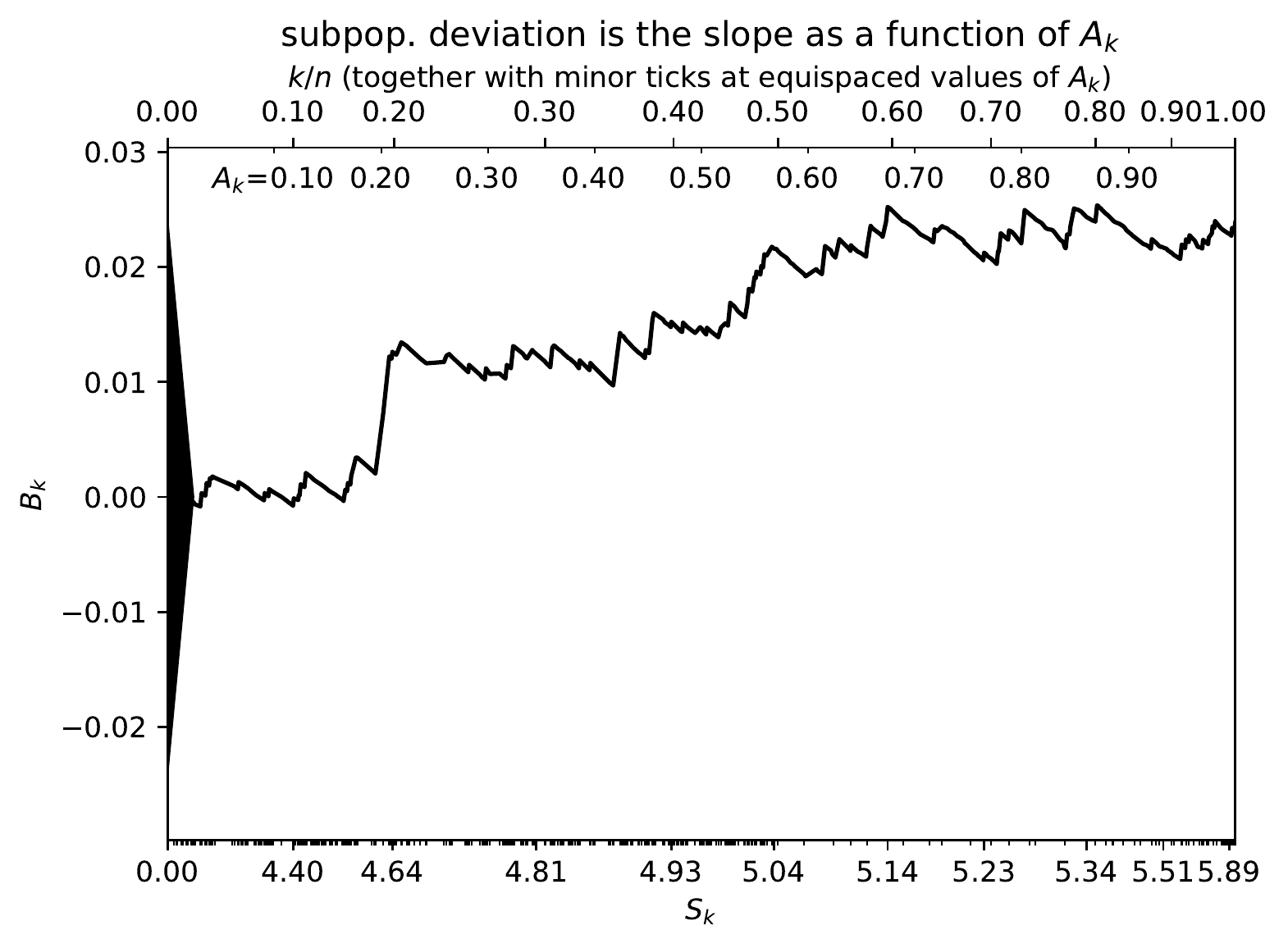}

\vspace{-.5em}

\ \ \ \ \ ($n =$ 506)
\end{center}
\vspace{-.125in}
\caption{
Difference in whether a household has internet access via satellite
between the county of Napa and the entire state of California
(the county is the subpopulation, while the state is the full population).
The scores indicated along the lower horizontal axis are $\log_{10}$
of the adjusted household income, randomly perturbed by about one part
in a hundred million to ensure their uniqueness.
There are 679 households representing Napa.
When the scores are perturbed at random ($n =$ 679),
Kuiper's statistic $H = 0.02761$, while $H/\sigma = 2.259$;
Kolmogorov's and Smirnov's $G = 0.02695$, while $G/\sigma = 2.205$.
When the responses are averaged for the same score ($n =$ 506),
Kuiper's statistic $H = 0.02619$, while $H/\sigma_n = 2.110$;
Kolmogorov's and Smirnov's $G = 0.02537$, while $G/\sigma_n = 2.043$.
The estimates of P-values for Kuiper's statistic are 0.0955 and 0.1392;
the estimates of P-values for Kolmogorov's and Smirnov's are 0.0549 and 0.0821.
The P-values reflect the observed difference
of not even a couple standard deviations beyond the expected means.
Deviation of the subpop.'s response (1 if satellite; 0 otherwise)
from the full population's is the slope.}
\label{napa}
\end{figure}

\section*{Acknowledgements}

We would like to thank Kamalika Chaudhuri, Imanol Arrieta Ibarra,
Michael Rabbat, Jonathan Tannen, and Susan Zhang.

\clearpage

\bibliographystyle{jasa3}

\bibliography{paper}

\end{document}